\documentclass[runningheads,hidelinks]{llncs}
\usepackage{algorithm}
\usepackage{amsfonts}
\usepackage{amsmath}
\usepackage{bussproofs}
\usepackage{etoolbox}
\usepackage{framed}
\usepackage{graphicx}
\usepackage{hyperref}
\usepackage[misc]{ifsym}
\usepackage{multirow}
\usepackage{pdflscape}
\usepackage{proof}
\usepackage{xcolor}
\usepackage{xspace,paralist}

\usepackage{vampire_induction}

\newbool{shortversion}
\boolfalse{shortversion}
\renewcommand{\paragraph}[1]{\par\smallskip\noindent\textbf{#1}}
\newcommand{\PH}[1]{{\color{magenta}#1}}

\newcommand{\reserved}[1]{\textbf{\underline{\texttt{#1}}}}
\definecolor{mygray}{gray}{0.9}
\newcommand{\newemph}[1]{%
\begingroup
\setlength{\fboxsep}{0pt}%
\colorbox{mygray}{\strut#1}%
\endgroup
}
\newcommand{\CommentedOut}[1]{}

\newtheorem{Theorem}[theorem]{Theorem}
\newtheorem{Corollary}[theorem]{Corollary}
\newtheorem{Lemma}[theorem]{Lemma}

\def\orcidID#1{\href{http://orcid.org/#1}{\raisebox{-1.25pt}{\includegraphics{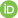}}}}

\begin{document}
\title{Program Synthesis in Saturation}
\author{Petra Hozzov\'a\inst{1}\orcidID{0000-0003-0845-5811} (\Letter) \and
Laura Kovács\inst{1}\orcidID{0000-0002-8299-2714} \and
Chase Norman\inst{2} \and
Andrei Voronkov\inst{3,4} \\
\email{petra.hozzova@tuwien.ac.at}
}
\authorrunning{Hozzov\'a et al.}

\institute{TU Wien \and UC Berkeley \and University of Manchester \and EasyChair}
\maketitle              % typeset the header of the contribution

%\setcounter{secnumdepth}{5}
% \setcounter{tocdepth}{5}
% \tableofcontents

\begin{abstract}
We present an automated reasoning framework for synthesizing recursion-free programs using saturation-based theorem proving.
Given a functional specification encoded as a first-order logical formula, we use a first-order theorem prover to both establish validity of this formula and discover program fragments satisfying the specification.
As a result, when deriving a proof of program correctness, we also synthesize a program that is correct with respect to the given specification. 
We describe properties of the calculus that a saturation-based prover capable of synthesis should employ, and extend the superposition calculus in a corresponding way.
We implemented our work  in the first-order prover \vampire{}, extending the successful applicability of first-order proving to program synthesis. 
\ifbool{shortversion}{}{

This is an extended version of an Automated Deduction -- CADE 29 paper with the same title and the same authors.}%
\end{abstract}
\section{Introduction}
%% PH: here I added the "we focus on..." to emphasize that we only consider synthesis given a functional specification
Program synthesis constructs code from a given specification.
In this work we focus on synthesis using functional specifications summarized by valid first-order formulas~\cite{MannaWaldinger1980,DBLP:series/natosec/AlurBDF0JKMMRSSSSTU15}, ensuring that our programs are provably correct.
%constructs code whose functional behavior is summarized by valid formulas.
%These formulas represent code specifications, implying
%These formulas imply that synthesis yields
%provably correct software.}
While being a powerful alternative to formal verification~\cite{DBLP:conf/popl/SrivastavaGF10}, %program 
synthesis faces  intrinsic computational challenges.
One of these challenges
is posed to the reasoning backend used for handling program specifications, as the latter typically include first-order quantifier alternations and interpreted theory symbols. As such, efficient reasoning with both theories and quantifiers is imperative for any effort towards program synthesis. 

%% PH: emphasized "recursion-free"
In this paper we address this demand for recursion-free programs. We advocate the use of first-order theorem proving  for extracting code from correctness proofs of  functional specifications %, with these proofs being derived in  saturation-based reasoning.
given as first-order formulas
%The  specifications  we consider are  first-order formulas
$\forall \overline{x} . \exists y. F[\overline{x},y]$.
These formulas state that   ``for all (program) inputs $\overline{x}$ there exists an output $y$ such that the input-output relation (program computation) $F[\overline{x}, y]$ is valid".
Given such a specification, we synthesize a recursion-free program while also deriving a proof certifying that the program satisfies the specification.

The programs we synthesize are built using first-order theory terms  extended with  $\itecons$  constructors. 
To ensure that our programs yield  computational models, i.e., that they can be evaluated for given values of input variables $\overline{x}$, we restrict the programs we synthesize to only contain  \emph{computable} symbols.

\paragraph{Our approach in a nutshell.} 
In order to synthesize a recursion-free program, we prove its functional specification using  saturation-based theorem proving~\cite{NieuwenhuisRubio:HandbookAR:paramodulation:2001,CAV13}. We extend saturation-based proof search with answer literals~\cite{Green69}, allowing us to track substitutions into the output variable $y$ of the specification.
These substitutions correspond to the sought program fragments %program fragments we aim to synthesize
and are conditioned on clauses they are associated with in the proof.
When we derive a clause corresponding to a program branch $\ifcmd\ C\ \thencmd\ r$, where $C$ is a condition and $r$ a term and both $C, r$ are computable, we %internally
store it and continue proof  search assuming that $\lnot C$ holds; we refer to such conditions %/clauses
$C$ as (program) branch conditions.
The saturation %-based
process for both proof search and code construction  terminates when the conjunction of negations of the collected branch conditions becomes unsatisfiable. Then we synthesize the final program satisfying the given (and proved) specification by assembling the recorded program branches (see e.g. Examples~\ref{ex:inv}-\ref{ex:inv2}). 
%
%\PH{Main challenges:}

The main challenges of making our approach effective come with 
(i) integrating the construction of the programs with $\itecons$ into the proof search, turning thus proof search into \emph{program search/synthesis},
and
(ii) guiding program synthesis to only computable branch conditions and programs.

%\PH{TODO -- contribution summary.
%What is new: saturation loop for program synthesis, semantics for clauses with answer literals, $\itecons$ in superposition, rules with abstraction, UWA algorithm for uncomputables, integration with AVATAR, theory reasoning.
%Note what isn't new: using answer literals for synthesis, $\itecons$ rules in general.}

\paragraph{Contributions.} We bring the next contributions solving the above challenges:\footnote{proofs of our results are given in
\ifbool{shortversion}{the extended version~\cite{VampirSyntPrePrint} of our paper}{Appendix~\ref{sec:appendix-proofs}}%
}\smallskip
\begin{compactitem}[\leftmargin=0.5mm]
    \item[$\bullet$]     We formalize the semantics for clauses with answer literals and  introduce a \emph{saturation-based algorithm for program synthesis} based on this semantics. We prove that, given a sound inference system, our saturation algorithm derives correct and computable programs  (Section~\ref{sec:synth-with-al}).
    %For efficiency and usability purposes, we also simplify the final program we synthesize (Section~\ref{sec:simplification}).
    %
   \item[$\bullet$] We define properties of a sound inference calculus in order to make the  calculus  suitable for our saturation-based algorithm for program synthesis. 
    We accordingly extend the superposition calculus and define a class of substitutions to be used within the extended calculus; we refer to these substitutions as  \emph{computable unifiers} (Section~\ref{sec:sup}). %We prove that the  extended calculus with computable unifiers has the desired properties (Section~\ref{sec:sup}).
    \item[$\bullet$] We extend a first-order unification algorithm  to find  computable unifiers (Section~\ref{sec:uwa}) to be further used in saturation-based program synthesis. % We also integrate our     saturation-based  algorithm with 
%computable unifiers with the AVATAR  framework~\cite{AVATAR14}, enabling  efficient splitting and theory reasoning (Section~\ref{sec:avatar}).

    %\item We describe the approaches we employ to simplify the final synthesized program (Section~\ref{sec:simplification}). %% ALREADY MENTIONED IN THE 1ST ITEM
   % \item We integrate our
    % saturation-based synthesis algorithm with the AVATAR  framework~\cite{AVATAR14} for efficient splitting and theory reasoning.  (Section~\ref{sec:avatar}).
    \item[$\bullet$] We implement our work in the \vampire{} prover~\cite{CAV13} and evaluate our synthesis approach on a number of examples, complementing other techniques in the area (Section~\ref{sec:impl:exp}). For example, our results demonstrate the applicability of our work on synthesizing programs for specifications that cannot be even encoded in the SyGuS syntax~\cite{sygus-standard}.
\end{compactitem}

%\PH{Should we mention somewhere here that our approach also supports reasoning with theories?}

\CommentedOut{
\PH{The outline/contributions distinction feels redundant to me, can we combine it into just one paragraph?}
This paper is organized as follows.
In Section~\ref{sec:preliminaries} we describe the preliminaries, including saturation, superpositon and answer literals.
In Section~\ref{sec:synth-with-al} we define semantics for clauses with answer literals and extend the saturation algorithm to program synthesis.
In Section~\ref{sec:sup} we present an extension of the superposition calculus for synthesis and illustrate it with two examples.
In Section~\ref{sec:uwa} we describe an algorithm for finding computable unifiers.
In Section~\ref{sec:simplification} we present a few simple methods for simplifying the synthesized programs.
In Section~\ref{sec:avatar} we describe integration of our synthesis method with AVATAR.
In Section~\ref{sec:experiments} we describe our implementation and show our practical results.
Finally, in Section~\ref{sec:related} we lay out the related work, before concluding in Section~\ref{sec:conclusions}.
}
\section{Preliminaries}\label{sec:preliminaries}
We assume familiarity with standard multi-sorted first-order logic with equality.
We denote variables by $x,y$, terms by $s, t$, atoms by $A$, literals by $L$, clauses by $C, D$, formulas by $F, G$, all possibly with indices. Further, we write $\sigma$ for Skolem constants. 
We reserve the symbol $\square$ for the \textit{empty clause}
which is logically equivalent to $\bot$.
Formulas and clauses with free variables are considered implicitly universally quantified (i.e. we consider closed formulas).
By $\simeq$ we denote the equality predicate and write $t\not\simeq s$ as a shorthand for $\lnot t\simeq s$.
We use a distinguished
\emph{integer sort}, denoted by~$\intg$. When we use standard integer predicates
$<$, $\leq$, $>$, $\geq$, functions $+, -, \dots$ and constants $0, 1, \dots$,
we assume that they denote the corresponding interpreted integer predicates and functions
with their standard interpretations.
Additionally, we include a conditional term constructor $\itecons$ in the language, as follows: given a formula $F$ and terms $s, t$ of the same sort, we write  $\ite{F}{s}{t}$ to denote the term $s$ if  % whose sort is as of $t$ and term of the same sort interpreted %in the standard way --
%as the interpretation of $s$ if the interpretation of 
$F$ is valid  and $t$ otherwise. 

 An \emph{expression} is a term, literal, clause or formula. 
We write $E[t]$ to denote that the expression $E$ contains the term $t$.
For simplicity, $E[s]$ denotes the expression $E$ where all occurrences of $t$ are replaced by the term $s$.
A \emph{substitution} $\theta$ is a mapping from variables to terms. A substitution $\theta$ is a \emph{unifier} of two expressions $E$ and $E'$ if $E\theta= E'\theta$, and is a \emph{most general unifier} (\emph{mgu}) if for every unifier $\eta$ of $E$ and $E'$, there exists substitution $\mu$ such that $\eta=\theta\mu$. We denote the mgu of $E$ and $E'$ with $\mathtt{mgu}(E,E')$.
%
%\PH{TODO: rewrite this to just have "distinguished" computable symbols. In the implementation we will consider the known interpreted functions to be computable + we will use annotations.}
%We use a set of closed formulas defining a \emph{theory} $T$.
We write $F_1,\ldots,F_n \vdash G_1,\ldots,G_m$ to denote
that $F_1 \land \ldots \land F_n \rightarrow G_1\lor\ldots\lor G_m$  is valid, and extend the notation also to validity modulo a theory $T$.  %holds in all
%models of $T$.
%We consider $T$ arbitrarily fixed and give all notions relative to $T$. For simplicity, we may drop the explicit reference to $T$.
Symbols occurring in a theory $T$  are 
\emph{interpreted} and all other symbols are \emph{uninterpreted}. 

% Every interpreted symbol is a \emph{computable} symbol.
% In addition, computable symbols include (i) distinguished functions and predicates provided by the user (e.g. library functions); (ii) the $\itecons$ function; and (iii) the skolem constants $\overline{\sigma}$ obtained by skolemization of input variables $\overline{x}$.
\subsection{Computable Symbols and Programs} 
We distinguish between \emph{computable} and \emph{uncomputable} symbols in the signature.
The set of computable symbols  is given as part of the specification.
Intuitively, a symbol is computable if it can be evaluated and hence is allowed to occur in a synthesized program.
A term or a literal is \emph{computable} if all symbols it contains are computable.
A symbol, term or literal is \emph{uncomputable} if it is not computable.
%We denote computable clauses by $C$ and uncomputable clauses by $U$ (both possibly with indices).
%A computable term $r[\overline{\sigma}]$ can be converted to a program $r[\overline{x}]$ 

A \emph{functional specification}, or simply just a \emph{specification},  is a formula 
\begin{equation}
    \forall \overline{x} . \exists y. F[\overline{x},y]. \label{eq:spec}
\end{equation}
%where $F$ is a first-order formula.
The variables $\overline{x}$ of a specification~\eqref{eq:spec} are called  \emph{input variables}.
Note that while we use specifications with a single variable $y$, our work can analogously be used with a tuple of variables $\overline{y}$ in~\eqref{eq:spec}.

Let $\overline{\sigma}$ denote a tuple of Skolem constants.  Consider a computable term $r[\overline{\sigma}]$ such that the instance $F[\overline{\sigma}, r[\overline{\sigma}]]$ of~\eqref{eq:spec} holds.
Since $\overline{\sigma}$ are fresh Skolem constants, the formula $\forall \overline{x} . F[\overline{x},r[\overline{x}]]$ also holds; we call such $r[\overline{x}]$ a \emph{program} for~\eqref{eq:spec} and say that the program $r[\overline{x}]$ \emph{computes a witness} of~\eqref{eq:spec}.

Further, if $\forall \overline{x}. (F_1\land\ldots\land F_n \rightarrow  F[\overline{x},r[\overline{x}]])$ holds for  computable formulas $F_1, \dots, F_n$, we write  $\prog{r[\overline{x}]}{\bigwedge_{i=1}^n F_i}$
to refer to  a \emph{program with conditions $F_1,\dots,\allowbreak F_n$} for~\eqref{eq:spec}. 
%As such, synthesizing a program with valid conditions $F_1,\dots,F_n$, with $F_i$ not containing any variables except for the input variables $\overline{x}$, corresponds to synthesizing a program computing the witness of~\eqref{eq:spec}.
In the sequel, 
we refer to (parts of) programs with conditions also as \emph{conditional branches}. 
%when $F_1,\dots,F_n$ hold.
%\PH{TODO: check the whole paper for how we use "program case", "conditional branch" and similar terms.}
%
In Section~\ref{sec:synth-with-al} we show how to build programs for~\eqref{eq:spec} by composing programs with conditions for~\eqref{eq:spec} (see Corollary~\ref{crl:2}).
% \begin{theorem}
%     If
%     \begin{align*}
%     T &\vdash \lnot F_1 \land \dots \land \lnot F_{i-1} \land F_i \rightarrow G[r_i] \qquad \text{for } i=1,\dots,k, \text{ and}\\
%     T &\vdash F_1 \lor \dots \lor F_k,
%     \end{align*}
%     then the program $P[\overline{x}]$ given as
% \begin{align}
% \begin{split}
% P[\overline{x}] :=\ &\ifcmd \ \lnot F_1[\overline{x}]\ \thencmd\ r_1[\overline{x}] \\
% &\elsecmd\ \ifcmd \ \lnot F_2[\overline{x}]\ \thencmd\ r_2[\overline{x}] \\
% &\qquad \dots \\
% &\elsecmd\ \ifcmd\ \lnot F_{k-1}[\overline{x}]\ \thencmd\ r_{k-1}[\overline{x}] \\
% &\elsecmd\ r_k[\overline{x}]
% \end{split}%\label{eq:prog-construction}
% \end{align}
% is a program for the specification $T\rightarrow \forall \overline{x}.\exists y. G[y]$.
% \end{theorem}

\subsection{Saturation and Superposition}
%\paragraph{Saturation.}
Saturation-based proof search implements \emph{proving by refutation}~\cite{CAV13}: to prove validity of $F$, a saturation algorithm establishes unsatisfiability of $\neg F$.
First-order theorem provers work with clauses, rather than with arbitrary formulas.
To prove a formula $F$, first-order provers negate $F$ which is further skolemized and converted to clausal normal form (CNF). The CNF of $\lnot F$ is denoted by  $\cnf(\lnot F)$ and represents a  set $S$ of initial clauses.
First-order provers then {\it saturate} $S$ by computing  logical consequences of $S$
 with respect to a sound inference system $\mathcal{I}$.
The saturated set of $S$ is called the {\it closure} of $S$
and the process of computing the closure of $S$ is called
\textit{saturation}. 
If the closure of $S$ contains the empty clause $\square$, the original
set $S$ of clauses is unsatisfiable, and hence the formula $F$ is valid.
%
% We show a simplified saturation algorithm for a sound inference system $\mathcal{I}$
% in Algorithm~\ref{alg:saturation}.
%% , with a clausified goal $F$ ($\neg F$ is also clausified) and clausified assumptions $A$ as input.

% \begin{algorithm}[t]
% \caption{The Saturation Loop.\label{alg:saturation}}
% \begin{tabbing}
% 	{\scriptsize 1}\quad initial set of clauses $S:=\{\cnf(\neg F)\}$\\ %active:=\emptyset$\\
% 	{\scriptsize 2}\quad \reserved{\texttt{repeat}}\\
% 	{\scriptsize 3}\quad\quad Select  clause $G\in S$\\
% 	{\scriptsize 4}\quad\quad Derive consequences ${C_1,\ldots,C_n}$ of $G$ and formulas from $S$ using rules of $\mathcal{I}$\\
% 	{\scriptsize 5}\quad\quad $S:=S\cup \{C_1,\ldots,C_n\}$\\
% %	{\scriptsize 6}\quad\quad $active := active\cup\{G\}$\\
% 	{\scriptsize 6}\quad\quad \reserved{\texttt{if}} $\square\in S$ \reserved{\texttt{then}}~ \reserved{\texttt{return}}~ $F$ is valid\\
% 	{\scriptsize 8}\quad \reserved{\texttt{return}}~ $F$ is not valid
% \end{tabbing}
% \end{algorithm}
%
 We may extend the set $S$ of initial clauses with additional clauses $C_1,\dots,C_n$.
If  $C$ is derived by saturating this extended set, we say $C$ is derived from $S$ \emph{under additional assumptions} $C_1,\dots,C_n$.

\begin{figure*}[t]
\centering
\begin{framed}
\vspace*{0.5em}
\textbf{Superposition:}\\
\begin{minipage}{0.95\textwidth}
\begin{equation*}
	\infer[]{(L[t]\lor C \lor C^\prime)\theta}{\underline{s\simeq t}\lor C\quad \underline{L[s^\prime]}\lor C^\prime} \quad
	\infer[]{(u[t]\not\simeq u^\prime\lor C \lor C^\prime)\theta}{\underline{s\simeq t}\lor C\quad \underline{u[s^\prime]\not\simeq u^\prime}\lor C^\prime} \quad
	\infer[]{(u[t]\simeq u^\prime\lor C \lor C^\prime)\theta}{\underline{s\simeq t}\lor C\quad \underline{u[s^\prime]\simeq u^\prime}\lor C^\prime}
\end{equation*}
where $\theta:=\mathtt{mgu}(s,s^\prime)$; $t\theta\not\succeq s\theta$; (first rule only) $L[s^\prime]$ is not an equality literal; and (second and third rules only) $u^\prime\theta\not\succeq u[s^\prime]\theta$.
\end{minipage}
\\[0.75em]
%We make use of the following inference rules of \Sup in this paper:\\
\begin{minipage}{.25\textwidth}%
\centering%
\textbf{Binary resolution:}
\begin{equation*}
	\infer[]{(C \lor C^\prime)\theta}{\underline{A}\lor C\quad\neg \underline{A^\prime}\lor C^\prime}
\end{equation*}
where\\ $\theta:=\mathtt{mgu}(A,A^\prime)$.\end{minipage}
\begin{minipage}{.01\textwidth}
\phantom{.}
\end{minipage}\begin{minipage}{.17\textwidth}
\centering
\!\!\!\textbf{Factoring:}
\begin{equation*}
	\infer[]{(A \lor C)\theta}{\underline{A}\lor \underline{A^\prime} \lor C}
\end{equation*}
where\\ $\theta\!:=\!\mathtt{mgu}(\!A,A^\prime\!)$.\end{minipage}
\begin{minipage}{.01\textwidth}
\phantom{.}
\end{minipage}\begin{minipage}{.27\textwidth}
\centering\vspace*{-0.3em}
\!\!\textbf{Equality resolution:}
\begin{equation*}
	\infer[]{C\theta}{\underline{s\not\simeq t}\lor C}
\end{equation*}
where $\theta:=\mathtt{mgu}(s,t)$.\\ \phantom{a}
\end{minipage}\begin{minipage}{.01\textwidth}
\phantom{.}
\end{minipage}\begin{minipage}{.25\textwidth}
\centering
\textbf{Equality factoring:}
\begin{equation*}
	\infer[]{(s\simeq t\lor t\not\simeq t' \lor C)\theta}{\underline{s\simeq t}\lor \underline{s'\simeq t'} \lor C}
\end{equation*}
where $\theta:=\mathtt{mgu}(s,s')$; $t\theta\not\succeq s\theta$; and $t'\theta\not\succeq t\theta$.
\end{minipage}
\vspace*{0.5em}
\end{framed}
\vspace*{-0.5em}
\caption{The superposition calculus \Sup{}. \label{fig:sup}}
\end{figure*}

%\paragraph{Superposition.}
The \textit{superposition calculus}, denoted as \Sup and given in Figure~\ref{fig:sup}, is the  most common inference system  used by saturation-based provers for  first-order logic with equality~\cite{NieuwenhuisRubio:HandbookAR:paramodulation:2001}. 
%An overview of its inference rules is given in Figure~\ref{fig:sup}.
The \Sup calculus is parametrized by a \emph{simplification ordering} $\succ$ on terms and a \emph{selection function}, which selects in each non-empty clause a non-empty subset of literals (possibly also positive literals).
We denote selected literals by underlining them.
An inference rule can be applied on the given premise(s) if the literals that are underlined in the rule are also selected in the premise(s).
For a certain class of selection functions, the superposition calculus \Sup{} is \textit{sound} (if $\square$ is derived from $F$, then $F$ is unsatisfiable) and
\textit{refutationally complete} (if $F$ is unsatisfiable, then $\square$ can be derived from it).

%For more details on saturation and superposition see~\cite{CAV13}.
\subsection{Answer Literals}\label{sec:anslits}
%In this paper we use answer literals to construct programs with conditions within saturation-based proof search.
%
{Answer literals}~\cite{Green69} provide a question answering technique for tracking substitutions into given variables throughout the proof.
Suppose we want to find a witness for the validity of the formula
\begin{equation}
    \exists y.F[y].\label{eq:question}
\end{equation}
Within  saturation-based proving, we first derive the skolemized negation of~\eqref{eq:question} and  add an \emph{answer literal} using a fresh predicate $\ans$ with argument  $y$, yielding
\begin{equation}
    \forall y.(\lnot F[y]\lor \ans(y)). \label{eq:qans}
\end{equation}
We then saturate the CNF of~\eqref{eq:qans}, while ensuring that answer literals are not selected for performing inferences. 
If the clause $\ans(t_1)\lor\ldots\lor\ans(t_m)$ is derived during saturation, note that this clause contains only answer literals in addition to the empty clause; hence, in this case we proved unsatisfiability of $\forall y.\lnot F[y]$, implying validity of~\eqref{eq:question}. Moreover,  $t_1,\dots,t_m$ provides a \emph{disjuntive answer}, i.e. witness,  for the validity of~\eqref{eq:question};  that is,    $F[t_1]\lor\ldots\lor F[t_m]$ holds~\cite{kunen1996semantics}.
In particular, if we derive the clause $\ans(t)$ during saturation, we
found a \emph{definite answer} $t$ for~\eqref{eq:question}, namely
$F[t]$ is valid.

\paragraph{Answer literals with $\itecons$.}\label{subsec:anslit-ite}
The derivation of disjunctive answers can be avoided by modifying the
inference rules to only derive clauses containing at most one answer literal.
One such modification is given within the $\text{A}(R)$-calculus for
binary resolution~\cite{Tammet1994}, where $R$ is a so-called strongly liftable term
restriction. The $\text{A}(R)$-calculus replaces the binary resolution
rule when both premises contain an answer literal by the following $A$-resolution rule:
%Here we present the simplified version of the rule:
\begin{equation*}
	%\vcenter{\infer[(A\text{-resolution})]{(D \lor D^\prime \lor \ans(\ite{A}{\overline{r^\prime}}{\overline{r}}))\theta}{A\lor D\lor \ans(\overline{r})\quad\neg A^\prime\lor D^\prime\lor \ans(\overline{r^\prime})}},
	\vcenter{\infer[(A\text{-resolution})]{(C \lor C^\prime \lor \ans(\ite{A}{r^\prime}{r}))\theta}{A\lor C\lor \ans(r)\quad\neg A^\prime\lor C^\prime\lor \ans(r^\prime)}},
\end{equation*}
where $\theta:=\mathtt{mgu}(A,A^\prime)$ and the restriction $R(\ite{A}{r^\prime}{r})$ holds.

In our work we go beyond the $A$-resolution rule and modify both the superposition calculus and the saturation algorithm to reason not only about answer literals but also about their use of  $\itecons$ terms (see Sections~\ref{sec:synth-with-al}--\ref{sec:sup}).
%Further, we go beyond extending the rules of the calculus, and present a modification of the saturation algorithm that allows us to construct programs while keeping the impact to the practical proof search efficiency low, as described in the following section. % Section~\ref{sec:synth-with-al}.

\section{Illustrative Example}\label{sec:motivation}
%With the provided necessary  background information, 
Let us illustrate our approach to program synthesis. We  
use answer literals in saturation to construct programs with conditions while proving specifications~\eqref{eq:spec}.
By adding an answer literal to the skolemized negation of~\eqref{eq:spec}, we obtain
\begin{equation}
    \forall y. (\lnot F[\overline{\sigma},y] \lor\ans(y)),
\end{equation}
where $\overline{\sigma}$ are the skolemized input variables $x$. When we derive a unit clause $\ans(r[\overline{\sigma}])$ during saturation, where $r[\overline{\sigma}]$ is a computable term, we construct a program for~\eqref{eq:spec} from the definite answer $r[\overline{\sigma}]$ by replacing $\overline{\sigma}$ with the input variables $\overline{x}$, obtaining the program $r[\overline{x}]$.
Hence, deriving computable definite answers by saturation allows us to synthesize  programs for specifications.

\begin{figure}[t]
    \centering
\begin{minipage}{0.25\linewidth}
\begin{equation}
\forall x.\, i(x)* x \simeq e \tag{A1}\label{eq:axA1}
\end{equation}
\end{minipage}
\begin{minipage}{0.04\linewidth}\,\end{minipage}
\begin{minipage}{0.22\linewidth}
\begin{equation}
\forall x.\, e* x \simeq x \tag{A2}\label{eq:axA2}
\end{equation}
\end{minipage}
\begin{minipage}{0.04\linewidth}\,\end{minipage}
\begin{minipage}{0.4\linewidth}
\begin{equation}
\forall x, y, z.\, x * (y * z) \simeq (x * y) * z \tag{A3}\label{eq:axA3}
\end{equation}
\end{minipage}
%     \begin{gather}
% \forall x.i(x)* x = e \tag{A1}\label{eq:axA1}\\
% \forall x.e* x = x \tag{A2}\label{eq:axA2}\\
% \forall x, y, z. x * (y * z) = (x * y) * z \tag{A3}\label{eq:axA3}
% \end{gather}
    \caption{Axioms defining a group. Uninterpreted function symbols $i(\cdot), e, *$ represent the inverse, the identity element, and the group operation, respectively.}
    \label{fig:groupax}
\end{figure}

\begin{example}\label{ex:inv}
Consider the group theory axioms (A1)--(A3) of Figure~\ref{fig:groupax}. We are interested in synthesizing a program for the following specification:
\begin{equation}
    \forall x.\exists y.\ x*y \simeq e \label{eq:spec-inv}
\end{equation}
% specification asserting that any  the right inverse element for any element $x$:
%the right inverse of $x$, we 
%
In this example we assume that all symbols %occurring in the axioms and the specification 
are computable.
To synthesize a program for~\eqref{eq:spec-inv},  
we add an answer literal to the skolemized negation
of~\eqref{eq:spec-inv} and convert the resulting formula to CNF (preprocessing).
We consider the set $S$ of clauses containing the obtained CNF and the axioms (A1)-(A3). We saturate $S$ using \Sup and obtain the following derivation:\footnote{For each formula in the derivation, we also list how the  formula has been derived. For example, formula~5 is the result of superposition (Sup) with formula~4 and axiom~A1, whereas binary resolution (BR) has been used to derive formula~6 from~5 and~1.}
\begin{enumerate}
\item \(\sigma*y\not\simeq e \lor \ans(y)\) \hfill [preprocessed specification]
\item \(i(x)*(x*y) \simeq e*y\) \hfill [Sup A1, A3]  %op(inv(X2),op(X2,X3)) = op(e,X3)
\item \(i(x)*(x*y) \simeq y\) \hfill [Sup A2, 2.] %op(inv(X2),op(X2,X3)) = X3
\item \(x*y \simeq i(i(x))*y\) \hfill [Sup 3., 3.] %op(X5,X6) = op(inv(inv(X5)),X6)
\item \(e \simeq x*i(x)\) \hfill [Sup 4., A1] %e = op(X0,inv(X0))
\item \(\ans(i(\sigma))\) \hfill [BR 5., 1.]
\end{enumerate}
Using the above derivation, we construct a program for the functional
specification~\eqref{eq:spec-inv} as follows: we replace $\sigma$
in the definite answer $i(\sigma)$ by $x$, yielding the program $i(x)$. % as the final synthesized program.
Note that for each input $x$, our
synthesized program
computes the inverse $i(x)$ of $x$ as an output.
In other words, our synthesized program for~\eqref{eq:spec-inv}
ensures that each group element $x$ has a right inverse $i(x)$. 
% specification asserting that any  the right inverse element for any element $x$:
%the right inverse of $x$, we 
% Therefore we found a definite answer $i(\sigma)$, i.e., we know that $\sigma*i(\sigma) \simeq e$ holds.
% Since $\sigma$ was introduced as a fresh constant symbol, it follows that also $\forall x. x*i(x) \simeq e$ holds, and $i(x)$ is the final synthesized program. 
\end{example}

While Example~\ref{ex:inv} yields a definite answer within saturation-based proof search, our work supports the synthesis of more complex recursion-free programs (see Examples~\ref{ex:group-ite}--\ref{ex:inv2}) by composing program fragments derived in the program search (Section~\ref{sec:synth-with-al}) as well as by using answer literals with $\itecons$ to effectively handle disjunctive answers (Section~\ref{sec:sup}). 

\section{Program Synthesis with Answer Literals}\label{sec:synth-with-al}
%\subsection{Finding Programs with Conditions}
% To use answer literals for program synthesis, we need to derive clauses that contain at most one answer literal with only computable arguments.
% To satisfy this requirement, we modify the saturation algorithm (as described in this section) as well as the superposition calculus (see Sections~\ref{sec:sup},~\ref{sec:uwa}).

We now introduce our approach to saturation-based program synthesis using answer literals (Algorithm~\ref{alg:saturation-new}). We focus on recursion-free program synthesis and present our work in a more general setting. Namely, we consider  functional specifications whose validity may depend on additional assumptions (e.g. additional program requirements) $A_1,\dots,A_n$, where each $A_i$ is a closed formula:
\begin{equation}
A_1\land\ldots\land A_n\rightarrow \forall \overline{x}.\exists y. F[\overline{x}, y] \label{eq:spec2}
\end{equation}
Note that specification~\eqref{eq:spec} is a special case of~\eqref{eq:spec2}.
However, since $A_1,\dots,A_n$ are closed formulas, \eqref{eq:spec2}~is equivalent to $\forall \overline{x}.\exists y.(A_1\land\ldots\land A_n\rightarrow  F[\overline{x}, y])$, which is a special case of~\eqref{eq:spec}.

Given a functional specification~\eqref{eq:spec2}, 
we use answer literals to synthesize programs with conditions (Section~\ref{sec:synt:pgmCond}) and extend saturation-based proof search to reason about answer literals (Section~\ref{sec:synth:SatAL}).  For doing so, 
we add  the answer literal $\ans(y)$ to the  skolemized negation of~\eqref{eq:spec2} and obtain
\begin{equation}
A_1\land\ldots\land A_n\land \forall y. (\lnot F[\overline{\sigma}, y] \lor\ans(y)).\label{eq:spec2ans}
\end{equation}
We  saturate the CNF of~\eqref{eq:spec2ans}, while ensuring that  answer literals are not selected within the inference rules used in saturation. We guide saturation-based proof search to derive clauses $C[\overline{\sigma}]\lor\ans(r[\overline{\sigma}])$, where $C[\overline{\sigma}]$ and $r[\overline{\sigma}]$ are computable.    

\subsection{From Answer Literals to Programs}\label{sec:synt:pgmCond}
Our next result 
ensures that, if we derive the clause
$C[\overline{\sigma}]\lor\ans(r[\overline{\sigma}])$, the term $r[\overline{\sigma}]$ is a definite answer under the assumption $\lnot C[\overline{\sigma}]$ (Theorem~\ref{thm:1}).
We note that we do not terminate saturation-based program synthesis once a clause  $C[\overline{\sigma}]\lor\ans(r[\overline{\sigma}])$ is derived. 
We rather record the program $r[\overline{x}]$ with condition $\lnot C[\overline{x}]$ (and possibly also other conditions), replace clause  $C[\overline{\sigma}]\lor\ans(r[\overline{\sigma}])$ by $C[\overline{\sigma}]$, and continue  saturation (Corollary~\ref{crl:1}). As a result, upon establishing validity of~\eqref{eq:spec2}, we synthesized a program for~\eqref{eq:spec2} (Corollary~\ref{crl:2}).

%The following theorem and its corollary extend the answer literal semantics as defined in~\cite{kunen1996semantics} and provide a basis for our method.
%The proofs can be found in the appendix.
%\PH{I'd propose rather calling the following theorem "Semantics of Clauses with Answer Literals". What do you think?}
\newcommand\theoremone{{\textnormal{\textbf{[Semantics of Clauses with  Answer Literals]}}}
Let $C$ be a clause not containing an answer literal. Assume that, using a saturation algorithm based on a sound inference system $\mathcal{I}$, the clause  $C\lor \ans(r[\overline{\sigma}])$  is derived from  the set of clauses consisting of  initial assumptions $A_1,\dots,A_n$, the clausified formula %\forall y.
$\cnf(\lnot F[\overline{\sigma}, y]\lor\ans(y))$ and additional assumptions $C_1,\dots,C_m$. Then, 
\[A_1,\dots,A_n,C_1,\dots,C_m\vdash C, F[\overline{\sigma}, r[\overline{\sigma}]].\]
That is, under the assumptions $C_1, \ldots, C_m,\lnot C$, the computable term $r[\overline{\sigma}]$ provides a definite answer to~\eqref{eq:spec2}. 
} 
\begin{Theorem}\label{thm:1}
\theoremone
\end{Theorem}
We further use Theorem~\ref{thm:1} to synthesize programs with conditions for~\eqref{eq:spec2}.
\newcommand\corollarytwo{{\textnormal{\textbf{[Programs with Conditions]}}}
Let $r[\overline{\sigma}]$ be a computable term and $C[\overline{\sigma}]$  a ground computable clause not containing an answer literal.   
Assume that clause $C[\overline{\sigma}]\lor \ans(r[\overline{\sigma}])$ is derived from the set of initial clauses $A_1,\dots,A_n$, the clausified formula 
%\forall y.
$\cnf(\lnot F[\overline{\sigma}, y]\lor\ans(y))$ and additional ground computable assumptions $C_1[\overline{\sigma}],\dots,C_m[\overline{\sigma}]$, by using saturation based on a sound inference system $\mathcal{I}$. Then, 
%\[A_1,\dots,A_n,C_1[\overline{\sigma}],\dots,C_m[\overline{\sigma}], \lnot C[\overline{\sigma}]\vdash F[\overline{\sigma}, r[\overline{\sigma}]].\]
%Further, if $r[\overline{\sigma}]$ is computable as well, then
$$\prog{r[\overline{x}]}{\bigwedge_{j=1}^{m}C_j[\overline{x}]\land\lnot C[\overline{x}]}$$
 is a program  with conditions for~\eqref{eq:spec2}.
%$C_1[\overline{x}]\land\dots\land C_m[\overline{x}]\land\lnot C[\overline{x}]$.
}
\begin{Corollary}\label{crl:1}
\corollarytwo
\end{Corollary}

%When all of $C_1[\overline{\sigma}],\dots,C_m[\overline{\sigma}],C[\overline{\sigma}]$ and $r[\overline{\sigma}]$ are computable, deriving $C[\overline{\sigma}]\lor ans(r[\overline{\sigma}])$ corresponds to synthesizing a part of a program computing the answer for the case where $C_1[\overline{x}],\dots,C_m[\overline{x}],\lnot C[\overline{x}]$ holds.
Note that a program with conditions $\prog{r[\overline{x}]}{\bigwedge_{j=1}^{m}C_j[\overline{x}]\land\lnot C[\overline{x}]}$ corresponds to a conditional (program) branch 
\(\ifcmd\ \bigwedge_{j=1}^{m}C_j[\overline{x}]\land\lnot C[\overline{x}]\ \thencmd\ r[\overline{x}]\): only if  the condition $\bigwedge_{j=1}^{m}C_j[\overline{x}]\land\lnot C[\overline{x}]$ is valid, then  $r[\overline{x}]$ is computed for~\eqref{eq:spec2}.  %, otherwise it does not compute anything.

We  use programs with conditions 
$\prog{r[\overline{x}]}{\bigwedge_{j=1}^{m}C_j[\overline{x}]\land\lnot C[\overline{x}]}$ to finally synthesize a program for~\eqref{eq:spec2}.
To this end, we use Corollary~\ref{crl:1} to derive programs with conditions, and once their conditions cover all possible cases given the initial assumptions $A_1,\dots,A_n$, we compose them into a program for~\eqref{eq:spec2}.
\newcommand\corollarythree{{\textnormal{\textbf{[From Programs with Conditions to Programs for~\eqref{eq:spec2}]}}}
Let $P_1[\overline{x}],\dots, P_k[\overline{x}]$, where $P_i[\overline{x}] = \prog{r_i[\overline{x}]}{\bigwedge_{j=1}^{i-1}C_j[\overline{x}]\land\lnot C_i[\overline{x}]}$, be programs with conditions for~\eqref{eq:spec2},
 such that $\bigwedge_{i=1}^n A_i\land\bigwedge_{i=1}^k C_i[\overline{x}]$ is unsatisfiable.
Then $P[\overline{x}]$, given by
\begin{align}
\begin{split}
P[\overline{x}] :=\ &\ifcmd \ \lnot C_1[\overline{x}]\ \thencmd\ r_1[\overline{x}] \\
&\elsecmd\ \ifcmd \ \lnot C_2[\overline{x}]\ \thencmd\ r_2[\overline{x}] \\
&\qquad \dots \\
&\elsecmd\ \ifcmd\ \lnot C_{k-1}[\overline{x}]\ \thencmd\ r_{k-1}[\overline{x}] \\
&\elsecmd\ r_k[\overline{x}],
\end{split}\label{eq:prog-construction}
\end{align}
is a program for~\eqref{eq:spec2}.
}
\begin{Corollary}\label{crl:2}
\corollarythree
\end{Corollary}
%\paragraph{From programs with conditions to programs for~\eqref{eq:spec2}.} 
%  To this end, we use Corollary~\eqref{crl:1} to derive  programs  $P_1[\overline{x}],\dots, P_k[\overline{x}]$ with conditions, where
% $P_i[\overline{x}] = \prog{r_i[\overline{x}]}{\bigwedge_{j=1}^{i-1}C_j[\overline{x}]\land\lnot C_i[\overline{x}]}$,
% such that $\bigwedge_{i=1}^n A_i\land\bigwedge_{i=1}^k C_i[\overline{x}]$ is unsatisfiable.
% The program for~\eqref{eq:spec2} is then given by 
% \begin{align}
% \begin{split}
% P[\overline{x}] :=\ &\ifcmd \ \lnot C_1[\overline{x}]\ \thencmd\ r_1[\overline{x}] \\
% &\elsecmd\ \ifcmd \ \lnot C_2[\overline{x}]\ \thencmd\ r_2[\overline{x}] \\
% &\qquad \dots \\
% &\elsecmd\ \ifcmd\ \lnot C_{k-1}[\overline{x}]\ \thencmd\ r_{k-1}[\overline{x}] \\
% &\elsecmd\ r_k[\overline{x}]
% \end{split}\label{eq:prog-construction}
% \end{align}
Note that since
the conditional branches of~\eqref{eq:prog-construction} cover all possible cases to be considered over $\overline{x}$, %and hence
we do not need the condition $\ifcmd\ \lnot C_k$. % 
%since the specification~\eqref{eq:spec2} assumes that $A_1,\dots,A_n$ hold, since $\bigwedge_{i=1}^n A_i \land\bigwedge_{i=1}^{k-1} C_i[\overline{x}] \rightarrow \lnot C_k[\overline{x}]$, and since $C_i[\overline{x}], r_i[\overline{x}]$ are computable for all $i$, the program $P[\overline{x}]$ computes the witness for~\eqref{eq:spec2} unconditionally.
%I.e., the conditional branches of $P[\overline{x}]$ cover all possible cases.
In particular, if $k=1$, i.e. $\bigwedge_{i=1}^n A_i\land C_1[\overline{x}]$ is unsatisfiable, then the synthesized program for~\eqref{eq:spec2} is  $r_1[\overline{x}]$. %(since $\bigwedge_{i=1}^n A_i\rightarrow \lnot C_1[\overline{x}]$).

%Therefore, in order to synthesize a program that computes the output value in any case, we aim to find witnesses with conditions such that for any input $x$, there is exactly one condition that holds.

% We synthesize parts of the program with such conditions, that the union of negations of all conditions is unsatisfiable with respect to $A_1,\dots,A_n$.
% We can then construct a program
% \begin{align*}
% \ifcmd &\ \lnot C_1[\overline{x}]\ \thencmd\ r_1[\overline{x}] \\
% \elsecmd\ \ifcmd &\ \lnot C_2[\overline{x}]\ \thencmd\ r_2[\overline{x}] \\
% & \dots \\
% \elsecmd\ \ifcmd &\ \lnot C_{k-1}[\overline{x}]\ \thencmd\ r_{k-1}[\overline{x}] \\
% \elsecmd\ &\ r_k[\overline{x}]
% \end{align*}

% In order to find such programs with conditions $P_i[\overline{x}]$, we need to modify the inference calculus so that it allows to derive clauses $C[\overline{\sigma}]\lor \ans(r[\overline{\sigma}])$, where both $C[\overline{\sigma}]$ and $r[\overline{\sigma}]$ are computable, while preserving the properties of Theorem~\ref{thm:1} and Corollary~\ref{crl:1}.
% That requires modifying the substitution used in inference rules such that uncomputable terms are not substituted into answer literals.
% We extend the superposition calculus and define properties of suitable substitution in Section~\ref{sec:sup}, and present a modification of the unification algorithm computing such substitution in Section~\ref{sec:uwa}.

\subsection{Saturation-Based Program Synthesis}\label{sec:synth:SatAL}
Our program synthesis results from Theorem~\ref{thm:1}, Corollary~\ref{crl:1} and Corollary~\ref{crl:2} 
rely upon a saturation algorithm using a sound (but not necessarily complete) inference system~$\mathcal{I}$.
In this section, we present our modifications to extend state-of-the-art saturation algorithms with answer literal reasoning, allowing to derive clauses  $C[\overline{\sigma}]\lor \ans(r[\overline{\sigma}])$, where both $C[\overline{\sigma}]$ and $r[\overline{\sigma}]$ are computable.
In Sections~\ref{sec:sup}--\ref{sec:uwa} we then describe modifications of the inference system $\mathcal{I}$ to implement  rules over clauses with answer literals. 

\begin{algorithm}[t]
\caption{Saturation Loop for Recursion-Free Program Synthesis  \label{alg:saturation-new}
%\PH{TODO: check whether we lost completeness by introducing the restrictions on uncomputables (and possibly uncomment the last line of the algo).}
}
\begin{tabbing}
	{\scriptsize 1\phantom{1}}\quad initial set of clauses $S:=\{\cnf(A_1\land\ldots\land A_n\land \forall y. (\lnot F[\overline{\sigma}, y] \lor\ans(y)))\}$\\ 
        {\scriptsize 2\phantom{1}}\quad initial sets of additional assumptions $\mathcal{C}:=\emptyset$ and programs $\mathcal{P}:=\emptyset$\\
        %active:=\emptyset$\\
	{\scriptsize 3\phantom{1}}\quad \reserved{\texttt{repeat}}\\
	{\scriptsize 4\phantom{1}}\quad\quad Select  clause $G\in S$\\
	{\scriptsize 5\phantom{1}}\quad\quad Derive consequences ${C_1,\ldots,C_n}$ of $G$ and formulas from $S$ using rules of $\mathcal{I}$\\
	{\scriptsize 6\phantom{1}}\quad\quad \reserved{\texttt{for each}} $C_i$ \reserved{\texttt{do}}\\
        {\scriptsize 7\phantom{1}}\quad\quad\quad \reserved{{if}} $C_i = (C[\overline{\sigma}]\lor \ans(r[\overline{\sigma}]))$ and $C[\overline{\sigma}]$ is ground and computable \reserved{\texttt{then}}\\
        {\scriptsize 8\phantom{1}}\quad\quad\quad\quad $\mathcal{P}:=\mathcal{P}\cup\{\prog{r[\overline{x}]}{\bigwedge_{C'\in\mathcal{C}}C'\land\lnot C[\overline{x}]}\}$ \qquad\qquad\qquad\qquad\;\; /* Corollary~\ref{crl:1} */\\
        {\scriptsize 9\phantom{1}}\quad\quad\quad\quad $\mathcal{C} := \mathcal{C}\cup \{C[\overline{x}]\}$\\
        {\scriptsize 10}\quad\quad\quad\quad $C_i := C[\overline{\sigma}]$\\
	{\scriptsize 11}\quad\quad $S:=S\cup \{C_1,\ldots,C_n\}$\\
%	{\scriptsize 10}\quad\quad $active := active\cup\{G\}$\\
	{\scriptsize 12}\quad\quad \reserved{\texttt{if}} $\square\in S$ \reserved{\texttt{then}}\\
        {\scriptsize 13}\quad\quad\quad \reserved{\texttt{return}}~ program~\eqref{eq:prog-construction} for specification~\eqref{eq:spec2},  derived from $\mathcal{P}$ \; /* Corollary~\ref{crl:2} */
%	{\scriptsize 11}\quad \reserved{\texttt{return}}~ $F$ is not valid
\end{tabbing}%
\vspace*{-1em}
\end{algorithm}

Our saturation algorithm is given in Algorithm~\ref{alg:saturation-new}. In  a nutshell, we use Corollary~\ref{crl:1} to construct   programs from  clauses $C[\overline{\sigma}]\lor\ans(r[\overline{\sigma}])$ and replace clauses $C[\overline{\sigma}]\lor\ans(r[\overline{\sigma}])$  by $C[\overline{\sigma}]$ (lines 7--10 of Algorithm~\ref{alg:saturation-new}). 
The newly added computable assumptions $C[\overline{\sigma}]$ are used to  guide saturation towards deriving programs with conditions %$\prog{r[\overline{x}]}{F}$ where $F$ is a conjunction containing the conjunct $C[\overline{x}]$
where the conditions contain $C[\overline{x}]$; these programs with conditions are used for synthesizing programs for~\eqref{eq:spec2}, as given in Corollary~\ref{crl:2}. 

Compared to a standard saturation algorithm used in first-order theorem proving (e.g. lines~4--5 of Algorithm~\ref{alg:saturation-new}), Algorithm~\ref{alg:saturation-new} implements additional steps for processing newly derived clauses $C[\overline{\sigma}]\lor\ans(r[\overline{\sigma}])$ with answer literals (lines 6-10). As a result, 
Algorithm~\ref{alg:saturation-new} establishes not only the validity of the specification~\eqref{eq:spec2}
but also synthesizes a program (lines 12-13).
%Compared to Algorithm~\ref{alg:saturation} it has additional processing of newly derived clauses $C_i$ (lines 5 and 6) and a program construction step after deriving the empty clause (line 8). \PH{TODO: maybe comment on not having the "return F is not valid"?}
Throughout the algorithm, we maintain a set $\mathcal{P}$ of programs with conditions derived so far and a set $\mathcal{C}$ of additional assumptions.
For each new clause $C_i$, we check if it is in the form $C[\overline{\sigma}]\lor\ans(r[\overline{\sigma}])$ where $C[\overline{\sigma}]$ is ground and computable (line 7).
If yes, we construct a program  with conditions 
$\prog{r[\overline{x}]}{\bigwedge_{C'\in\mathcal{C}}C'\land\lnot C[\overline{x}]}$, extend $\mathcal{C}$ with the additional assumption $C[\overline{x}]$, and replace $C_i$ by $C[\overline{\sigma}]$ (lines 8-10).
Then, when we derive the empty clause, we construct the final program as follows.
We first collect all clauses that participated in the derivation of $\square$.
We use this clause collection to filter the programs in $\mathcal{P}$ -- we only keep a program originating from a clause $C[\overline{\sigma}]\lor\ans(r[\overline{\sigma}])$ if the condition $C[\overline{\sigma}]$ was used in the proof,
%(with the exception of the final clause $\square$),
obtaining programs $P_1,\dots,P_k$.
From $P_1,\dots,P_k$ we then synthesize the final program $P$ using the construction~\eqref{eq:prog-construction} from Corollary~\ref{crl:2}.
%Note that the condition to conclude the program search is deriving $\square$, not $\ans(r)$ (as in Subsection~\ref{sec:anslits}), because if we derive $\ans(r)$, it is immediately replaced by $\square$ (line 10).
%The correctness of our program synthesis algorithm, when combined with a sound inference system, follows from Corollary~\ref{crl:1}.

\begin{remark}\label{rem:SATComp}
Compared to~\cite{Tammet1994} where potentially large programs (with conditions) are tracked in answer literals, 
Algorithm~\ref{alg:saturation-new} removes answer literals from clauses and constructs the final program only after saturation found a refutation of the negated~\eqref{eq:spec2}. 
Our approach has two advantages: first, we do not have to keep track of potentially many large terms using $\itecons$, which might slow down saturation-based program synthesis. Second, our work can naturally be integrated with clause splitting techniques within saturation (see Section~\ref{sec:avatar}).
%However, the conditions $C[\overline{\sigma}]$ that our framework generates might be more complex then necessary.
%For further efficiency, in Section~\ref{sec:simplification} we present additional simplification over the conditions $C[\overline{\sigma}]$. 
% \PH{Do we need an example for this?
% E.g. the program generated for the max-of-2 example:
% \[\ite{x_1\geq x_2 \land x_1\geq x_1}{x_1}{x_2}\]}
\end{remark}

\section{Superposition with Answer Literals}\label{sec:sup}
%\PH{Note somewhere (in this section?) that this approach works also with theory reasoning.}

We note that our saturation-based program synthesis approach is not restricted to a specific calculus.  Algorithm~\ref{alg:saturation-new} can thus be used with \emph{any sound} set of inference rules, including theory-specific inference rules, e.g.~\cite{ALASCA23}, as long as the rules allow derivation of clauses in the form $C \lor \ans(r)$, where $C, r$ are computable and $C$ is ground.
I.e., the rules should only derive clauses with at most one answer literal, and should not introduce uncomputable symbols into answer literals.

In this section we  present changes tailored to the superposition calculus \Sup{}, yet, without changing the underlying saturation process of Algorithm~\ref{alg:saturation-new}. 
We first introduce the notion of an abstract unifier~\cite{UWA-THI} and define a computable unifier -- a mechanism for dealing with the uncomputable symbols in the reasoning instead of introducing them into the programs.
The use of such a unifier in any sound calculus is explained, with particular focus on the  \Sup{} calculus.

\begin{definition}[Abstract unifier~\cite{UWA-THI}]
An \emph{abstract unifier} of two expressions $E_1, E_2$ is a pair $(\theta, D)$ such that:
\begin{enumerate}
\item $\theta$ is a substitution and $D$ is a (possibly empty) disjunction of disequalities,
\item $(D\lor E_1\simeq E_2)\theta$ is valid in the underlying theory.
\end{enumerate}
\end{definition}
Intuitively speaking,  an abstract unifier combines disequality constraints $D$ with a substitution $\theta$ such that the substitution is a unifier of $E_1, E_2$ if the constraints $D$ are not satisfied.

\begin{definition}[Computable unifier]
  A \emph{computable unifier} of two expressions $E_1, E_2$ with respect to an expression $E_3$ is an abstract unifier $(\theta, D)$ of $E_1, E_2$ such that the expression $E_3\theta$ is computable.
\end{definition}

For example, let $f$ be computable and $g$ uncomputable.
Then $(\{y\mapsto f(z)\}, \allowbreak z \not\simeq g(x))$ is a computable unifier of the terms $f(g(x)),y$ with respect to $f(y)$.
Further, $(\{y\mapsto f(g(x))\}, \emptyset)$ is an abstract unifier of the same terms, but not a computable unifier with respect to $f(y)$.

\paragraph{Ensuring computability of answer literal arguments.}
%To ensure that the arguments of the answer literals remain computable throughout the program search,
We modify the rules of a sound inference system~$\mathcal{I}$ to use computable unifiers with respect to the answer literal argument instead of unifiers.
Since a computable unifier may entail disequality constraints $D$, we add $D$ to the  conclusions of the inference rules. That is, for an inference rule of $\mathcal{I}$ as below 
\begin{equation}
    \infer[,]{C\theta}{C_1 \quad \cdots \quad C_n}\label{eq:orig-rule}
\end{equation}
where $\theta$ is a substitution such that $E\theta\simeq E^\prime\theta$ holds for some expressions $E, E^\prime$, we extend $\mathcal{I}$ with the following $n$ inference rules with computable unifiers:
\begin{equation}
\infer[]{(\underline{D}\lor C\lor \ans(r))\theta'}{C_1 \lor \ans(r) \quad C_2 \quad \cdots \quad C_n}
\quad \cdots \quad
\infer[,]{(\underline{D}\lor C\lor \ans(r))\theta'}{C_1 \quad C_2 \quad \cdots \quad C_n \lor \ans(r)}\label{eq:new-rules}
\end{equation}
where $(\theta', D)$ is a computable unifier of $E, E^\prime$ with respect to $r$ and none of $C_1,\dots,C_n$ contains an answer literal.  We obtain the following result.

\newcommand\lemmathree{
{\textnormal{\textbf{[Soundness of Inferences with  Answer Literals]}}}
%Assume rule~\eqref{eq:orig-rule} is sound.
%Assume also that, whenever  the conclusion $C\theta$ of~\eqref{eq:orig-rule} is false in an interpretation $I$ under variable assignment $v$, then some of the premises  $C_1,\dots,C_n$ of~\eqref{eq:orig-rule} are also false in $I$ under variable assignment $v^\prime$, where $v^\prime$ assigns to each variable $x$ the value $x\theta$ has in $I$ under $v$.
If the rule~\eqref{eq:orig-rule} is sound,
the rules~\eqref{eq:new-rules} are sound as well.}
\begin{Lemma}\label{lemma:soundness1}
\lemmathree
\end{Lemma}

\begin{figure*}[t]
\centering
\begin{framed}
\vspace*{0.5em}
\textbf{Superposition (Sup):}\\
\begin{minipage}{0.95\textwidth}
{\scriptsize
\begin{gather*}
    \infer[]%[\!\!\!1]
    {(\underline{D}\lor L[t]\lor C \lor C^\prime\lor\ans(\ite{s\!\simeq\! t}{r^\prime}{r}))\theta}{\underline{s\simeq t}\lor C\lor \ans(r)\quad \underline{L[s^\prime]}\lor C^\prime\lor\ans(r^\prime)} \;\;
    \infer%[\!\!\!2]
    []{(\underline{D} \lor \underline{r\!\not\simeq\! r^\prime}\lor L[t]\lor C \lor C^\prime \lor \ans(r))\theta}
    {\underline{s\simeq t}\lor C\lor \ans(r)\quad \underline{L[s^\prime]}\lor C^\prime\lor \ans(r^\prime)} \\
    \infer%[\!\!\!3]
    []{(\underline{D}\!\lor\! u[t]\!\not\simeq\! u^\prime\!\lor\! C \!\lor\! C^\prime\!\lor\!\ans(\ite{\!s\!\simeq\! t\!}{\!r^\prime\!}{\!r}))\theta}{\underline{s\simeq t}\lor C\lor\ans(r)\quad \underline{u[s^\prime]\not\simeq u^\prime}\lor C^\prime\lor\ans(r^\prime)} \;\;
    \infer%[\!\!\!4]
    []{(\underline{D} \!\lor\! \underline{r\!\not\simeq\! r^\prime}\!\lor\! u[t]\!\simeq\! u^\prime\!\lor\! C \!\lor\! C^\prime \!\lor\! \ans(r))\theta}
    {\underline{s\!\simeq\! t}\!\lor\! C\!\lor\! \ans(r)\quad \underline{u[s^\prime]\!\simeq\! u^\prime}\!\lor\! C^\prime\!\lor\! \ans(r^\prime)}\\
    \infer%[\!\!\!5]
    []{(\underline{D}\!\lor\! u[t]\!\simeq\! u^\prime\!\lor\! C \!\lor\! C^\prime\!\lor\!\ans(\ite{\!s\!\simeq\! t\!}{\!r^\prime\!}{\!r\!}))\theta}{\underline{s\simeq t}\lor C\lor\ans(r)\quad \underline{u[s^\prime]\simeq u^\prime}\lor C^\prime\lor\ans(r^\prime)} \;\;
    \infer%[\!\!\!6]
    []{(\underline{D} \!\lor \!\underline{r\!\not\simeq\! r^\prime}\!\lor\! u[t]\!\not\simeq\! u^\prime\!\lor\! C \!\lor\! C^\prime \!\lor\! \ans(r))\theta}
    {\underline{s\!\simeq\! t}\!\lor\! C\!\lor\! \ans(r)\quad \underline{u[s^\prime]\!\not\simeq\! u^\prime}\!\lor\! C^\prime\!\lor\! \ans(r^\prime)}
\end{gather*}}%
where $(\theta, D)$ is a computable unifier of $s,s^\prime$ w.r.t. the argument of the answer literal in the rule conclusion (i.e. $\ite{s\!\simeq\! t}{r^\prime}{r}$ for the left-column rules, and $r$ for the others); (rules on the first line only) $L[s^\prime]$ is not an equality literal; and (rules on the second and third line only) $u^\prime\theta\not\succeq u[s^\prime]\theta$.
\end{minipage}
\\[0.75em]
\begin{minipage}{.95\textwidth}
{\centering
\textbf{Binary resolution (BR):}
\begin{equation*}
% VERSION OF THE CALCULUS WITH \overline{r}
	% \infer[]{(\underline{D}\lor C \lor C^\prime\lor \ans(\ite{A}{\overline{r^\prime}}{\overline{r}}))\theta}{A\lor C\lor \ans(\overline{r})\quad\neg A^\prime\lor C^\prime\lor \ans(\overline{r^\prime})} \;
 % \infer[]{(\underline{D}\lor \bigvee_{i} \underline{r_i\not\simeq r_i^\prime}\lor C \lor C^\prime \lor \ans(\overline{r}))\theta}
 %        {A\lor C\lor \ans(\overline{r})\quad \neg A^\prime\lor C^\prime\lor \ans(\overline{r^\prime})}
	\infer%[\!\!\!{\scriptsize\text{1}}]
        []{(\underline{D}\!\lor\! C \!\lor\! C^\prime\!\lor\! \ans(\ite{A}{r^\prime}{r}))\theta}{\underline{A}\lor C\lor \ans(r)\quad\neg \underline{A^\prime}\lor C^\prime\lor \ans(r^\prime)} \;\;\;
        \infer%[\!\!\!{\scriptsize\text{2}}]
        []{(\underline{D}\lor \underline{r\!\not\simeq\! r^\prime}\lor C \lor C^\prime \lor \ans(r))\theta}
        {\underline{A}\lor C\lor \ans(r)\quad \neg \underline{A^\prime}\lor C^\prime\lor \ans(r^\prime)}
\end{equation*}}%
where $(\theta, D)$ is a computable unifier of $A,A^\prime$ w.r.t. (first rule) $\ite{A}{r^\prime}{r}$ or (second rule) $r$. \end{minipage}
\\[0.75em]
%\begin{minipage}{.05\textwidth}\phantom{a}\end{minipage}
%\hspace*{0em}
\begin{minipage}{.23\textwidth}
\centering
\textbf{Factoring (F):}
\begin{equation*}
% VERSION OF THE CALCULUS WITH \overline{r}
	% \infer[]{(\underline{D}\lor A \lor C\lor \ans(\overline{r}))\theta}{A\lor A^\prime \lor C\lor \ans(\overline{r})}
	\infer[]{(\underline{D}\!\lor\! A \!\lor\! C\!\lor\! \ans(r))\theta}{\underline{A}\lor \underline{A^\prime} \lor C\lor \ans(r)}
\end{equation*}
where $(\theta, D)$ is a\\ computable unifier \\of $A,A^\prime$ w.r.t. $r$.\end{minipage}
%\\[0.75em]
\begin{minipage}{.01\textwidth}
\phantom{a}
\end{minipage}\begin{minipage}{.34\textwidth}
\centering
\textbf{Equality resolution\! (ER):}
\begin{equation*}
% VERSION OF THE CALCULUS WITH \overline{r}
	% \infer[]{(\underline{D}\lor C\lor \ans(\overline{r}))\theta}{s\not\simeq t\lor C\lor \ans(\overline{r})}
	\infer[]{(\underline{D}\lor C\lor \ans(r))\theta}{\underline{s\not\simeq t}\lor C\lor \ans(r)}
\end{equation*}
where $(\theta, D)$ is a\\ computable unifier\\ of $s,t$ w.r.t. $r$.
\end{minipage}\begin{minipage}{.02\textwidth}
\phantom{a}
\end{minipage}\begin{minipage}{.35\textwidth}
%\vspace*{0.3em}
\centering
\textbf{Equality factoring (EF):}
%\vspace*{-0.1em}
\begin{equation*}
% VERSION OF THE CALCULUS WITH \overline{r}
	% \infer[]{(\underline{D}\lor s\simeq t\lor t\not\simeq t' \lor C\lor \ans(\overline{r}))\theta}{s\simeq t\lor s'\simeq t' \lor C\lor \ans(\overline{r})}
	\infer[]{(\underline{D}\!\lor\! s\!\simeq\! t\!\lor\! t\!\not\simeq\! t' \!\lor\! C\!\lor\! \ans(r))\theta}{\underline{s\simeq t}\lor \underline{s'\simeq t'} \lor C\lor \ans(r)}
\end{equation*}
where $(\theta, D)$ is a computable\\ unifier of $s,s^\prime$ w.r.t. $r$; \\ $t\theta\not\succeq s\theta$; and $t'\theta\not\succeq t\theta$.
\end{minipage}
\vspace*{0.5em}
\end{framed}
\vspace*{-0.5em}
\caption{%Selected rules of the e
Selected rules of the extended superposition calculus \Sup{} for reasoning with answer literals\label{fig:sup-new}, with  underlined literals being selected.
%\PH{Question: should we even display the non-$\itecons$ versions of Sup and BR? All our examples can be solved without them (and we don't have them implemented).}%\\
% \PH{Note: I believe we don't need the symmetric versions of the non-$\itecons$ rules, since discharging the constraints $r\not\simeq r^\prime$ effectively unifies the answer literals $\ans(r),\ans(r^\prime)$. Check it though.}
}
\end{figure*}

We note that we keep the original rule~\eqref{eq:orig-rule} in $\mathcal{I}$, but impose that none of its premises $C_1,\dots,C_n$ contains an answer literal.
Clearly, neither the such modified rule~\eqref{eq:orig-rule} nor the new rules~\eqref{eq:new-rules} introduce uncomputable symbols into answer literals.
Rather, these rules add disequality constraints $D$ into their conclusions and immediately select $D$ for further applications of inference rules.
Such a selection guides the saturation process in Algorithm~\ref{alg:saturation-new} to first discharge the constraints $D$ containing uncomputable symbols with the aim of deriving a clause $C^\prime\lor\ans(r^\prime)$ where $C^\prime$ is computable. The clause $C^\prime\lor\ans(r^\prime)$ is then converted into a program with conditions using Corollary~\ref{crl:1}. 

\paragraph{Superposition with answer literals.} We make the inference rule modifications~\eqref{eq:orig-rule}, together with the addition of new rules~\eqref{eq:new-rules}, for each inference rule of the \Sup{} calculus from Figure~\ref{fig:sup}. 
Further, 
%\paragraph{Deriving clauses with at most one answer literal.}
%Moreover,
we also ensure that rules with multiple premises, when applied on several premises containing answer literals, \emph{derive clauses with at most one answer literal}.
We therefore introduce the following  two rule modifications. (i) We use the $\itecons$  constructor to combine  answer literals of  premises, by adapting the use of $\itecons$ within binary resolution ~\cite{LWC1974,MannaWaldinger1980,Tammet1994} to superposition rules.
(ii) We use an answer literal from only one of the rule premises in the rule conclusion and add new disequality constraint $r\not\simeq r'$ between the premises' answer literal arguments, similar to the constraints $D$ of the computable unifier.
Analogously to the computable unifier constraints, we immediately select  this disequality constraint $r\not\simeq r'$.

The resulting extension of the \Sup{} calculus with answer literals is given in Figure~\ref{fig:sup-new}.
In addition to the rules of Figure~\ref{fig:sup-new}, the extended calculus contains rules constructed as~\eqref{eq:new-rules} for superposition and binary resolution rules of Figure~\ref{fig:sup}.
Using Lemma~\ref{lemma:soundness1}, we  conclude the following. 

\newcommand\lemmafour{
{\textnormal{\textbf{[Soundness of \Sup with  Answer Literals]}}} The inference rules from Figure~\ref{fig:sup-new} of the extended \Sup{} calculus with answer literals are sound.}
\begin{Lemma}\label{lemma:soundness2}
\lemmafour
\end{Lemma}
%
%The following lemma states soundness of the new rules under certain conditions (fulfilled by the rules of factoring, equality resolution and equality factoring). Proof can be found in Appendix~\ref{sec:appendix-proofs}.

By the soundness results of Lemmas~\ref{lemma:soundness1}--\ref{lemma:soundness2}, Corollaries~\ref{crl:1}--\ref{crl:2} imply that, when applying the calculus of Figure~\ref{fig:sup-new} in the saturation-based program synthesis approach of Algorithm~\ref{alg:saturation-new}, we construct correct programs.

\begin{example}\label{ex:group-ite}
We illustrate the use of Algorithm~\ref{alg:saturation-new} with the extended \Sup{} calculus of Figure~\ref{fig:sup-new}, strengthening our motivation from Section~\ref{sec:motivation} with $\itecons$ reasoning. 
%
% \begin{equation}
%     \forall x.\ x*e\simeq x, \tag{A2'}
% \end{equation}
To this end, consider the functional specification over group theory:
\begin{equation}
    \forall x,y.\exists z.(x*y\not\simeq y*x \rightarrow z*z\not\simeq e), \label{eq:spec-ex2}
\end{equation}
asserting that, if the group is not commutative, there is an element whose square is not $e$. 
In addition to the 
 axioms (A1)-(A3) of Figure~\ref{fig:groupax}, we also use the right identity axiom (A2')~$\forall x.\ x*e\simeq x$.\footnote{We  include  axiom (A2') only to shorten the presentation of the obtained derivation.}
%
%
%
%describing an element $z$, square of which is $e$, given that the group operation is not commutative.
%
%\PH{TODO: remove the introduction of (A2') if we already introduce it in Example~\ref{ex:inv2}.}
Based on Algorithm~\ref{alg:saturation-new}, we obtain the following derivation of the program for~\eqref{eq:spec-ex2}: %using our framework
% WITHOUT ANSWER LITERALS, JUST THE PROOF:
% \begin{enumerate}
%     % 12. op(sK0,sK1) != op(sK1,sK0) [cnf transformation 11]; 1
%     \item $\sigma_1*\sigma_2\not\simeq \sigma_2 * \sigma_1$ \hfill [preprocessed specification]
%     % 13. e = op(X2,X2) [cnf transformation 11]; 2
%     \item $e \simeq x*x$ \hfill [preprocessed specification]
%     % AXIOMS: we do not write them in the proof, they're in the Figure
%     % 14. op(op(X1,X0),X2) = op(X1,op(X0,X2)) [cnf transformation 8]; A3
%     % 15. op(e,X0) = X0 [cnf transformation 2]; A2
%     % END AXIOMS
%     % 17. op(X0,op(X0,X1)) = op(e,X1) [superposition 14,13]; 3
%     \item $x*(x*y) \simeq e*y$ \hfill [Sup A3, 2.]
%     % 21. e = op(X0,op(X1,op(X0,X1))) [superposition 14,13]; 4
%     \item $e \simeq x*(y*(x*y))$ \hfill [Sup A3, 2.]
%     % 23. op(X0,op(X0,X1)) = X1 [forward demodulation 17,15]; 5
%     \item $x*(x*y) \simeq y$ \hfill [Sup 3., A2]
%     % 27. op(X0,e) = X0 [superposition 23,13]; 6
%     \item $x*e \simeq x$ \hfill [Sup 5., 2.]
%     % 70. op(X2,e) = op(X3,op(X2,X3)) [superposition 23,21]; 7
%     \item $x*e \simeq y*(x*y)$ \hfill [Sup 5., 4.]
%     % 81. op(X3,op(X2,X3)) = X2 [forward demodulation 70,27]; 8
%     \item $y*(x*y) \simeq x$ \hfill [Sup 7., 6.]
%     % 97. op(X2,X3) = op(X3,X2) [superposition 23,81]; 9
%     \item $x*y \simeq y*x$ \hfill [Sup 5., 8.]
%     % 240. op(sK0,sK1) != op(sK0,sK1) [superposition 12,97]; 10
%     \item $\square$ \hfill [BR 9., 1.]
% \end{enumerate}

% WITH ANSWER LITERALS:
{
\begin{enumerate}
    % 12. op(sK0,sK1) != op(sK1,sK0) [cnf transformation 11]; 1
    \item $\sigma_1*\sigma_2\not\simeq \sigma_2 * \sigma_1 \lor \ans(z)$ \hfill [preprocessed specification]
    % 13. e = op(X2,X2) [cnf transformation 11]; 2
    \item $e \simeq z*z \lor \ans(z)$ \hfill [preprocessed specification]
    % 12. op(sK0,sK1) != op(sK1,sK0) [cnf transformation 11]; 3
    \item $\sigma_1*\sigma_2\not\simeq \sigma_2 * \sigma_1$ \hfill [answer literal removal 1. (Algorithm~\ref{alg:saturation-new}, line 10)]
    % AXIOMS: we do not write them in the proof, they're in the Figure
    % 14. op(op(X1,X0),X2) = op(X1,op(X0,X2)) [cnf transformation 8]; A3
    % 15. op(e,X0) = X0 [cnf transformation 2]; A2
    % END AXIOMS
    % 17. op(X0,op(X0,X1)) = op(e,X1) [superposition 14,13]; 4
    \item $x*(x*y) \simeq e*y \lor\ans(x)$ \hfill [Sup 2., A3]
    % 21. e = op(X0,op(X1,op(X0,X1))) [superposition 14,13]; 5
    \item $e \simeq x*(y*(x*y))\lor\ans(x*y)$ \hfill [Sup A3, 2.]
    % 23. op(X0,op(X0,X1)) = X1 [forward demodulation 17,15]; 6
    \item $x*(x*y) \simeq y\lor\ans(x)$ \hfill [Sup 4., A2]
 %   % 27. op(X0,e) = X0 [superposition 23,13]; 7
 %   \item $x*e \simeq x\lor\ans(x)$ \hfill [Sup 6., 2.]
    % 70. op(X2,e) = op(X3,op(X2,X3)) [superposition 23,21]; 7
    \item $x*e \simeq y*(x*y)\lor\ans(\ite{e \simeq x*(y*(x*y))}{x}{x*y})$ \hfill [Sup 6., 5.]
    % 81. op(X3,op(X2,X3)) = X2 [forward demodulation 70,27]; 8
    \item $y*(x*y) \simeq x\lor\ans(\ite{e \simeq x*(y*(x*y))}{x}{x*y})$ \hfill [Sup 7., A2']
    % 97. op(X2,X3) = op(X3,X2) [superposition 23,81]; 9
    \item $x*y \simeq y*x \lor\ans(\ite{x*(y*x)\simeq y}{x}{\ite{e \simeq x*(y*(x*y))}{x}{x*y}})$ \hfill [Sup 6., 8.]
    % 240. op(sK0,sK1) != op(sK0,sK1) [superposition 12,97]; 11
    \item $\ans(\ite{\sigma_1*(\sigma_2*\sigma_1)\simeq \sigma_2}{\sigma_1}{\ite{e \simeq \sigma_1*(\sigma_2*(\sigma_1*\sigma_2))}{\sigma_1}{\\\sigma_1*\sigma_2}})$  \hfill [BR 9., 3.]
    \item $\square$ \hfill [answer literal removal 11.  (Algorithm~\ref{alg:saturation-new}, line 10)]
\end{enumerate}}%
The programs with conditions collected during saturation-based program synthesis, in particular corresponding to steps 3. and 11. above, are:
{\small
\begin{align*}
P_1[x,y] :=&\ \langle z, x*y \simeq y*x \rangle \\
P_2[x,y] :=&\ \langle \ite{x*(y*x)\simeq y}{x}{(\ite{e \simeq x*(y*(x*y))}{x}{x*y})}, \\
    &\;\; x*y \not\simeq y*x \rangle
%P_2[x,y] :=&\ \langle \ifcmd\ x*(y*x)\simeq y\ \thencmd\ x\\
%    &\qquad\qquad \elsecmd\ \ite{e \simeq x*(y*(x*y))}{x}{x*y}, \\
%    &\;\; x*y \not\simeq y*x \rangle
\end{align*}}%
% From $P_1[x,y], P_2[x,y]$ we could construct the following program:
%     \begin{align*}
%         &\ \ifcmd\  x*y \simeq y*x\ \thencmd\ z\\
%         &\qquad\qquad \elsecmd\ \ifcmd\ x*(y*x)\simeq x\ \thencmd\ x\\
%         &\qquad\qquad\qquad\qquad \elsecmd\ \ite{e=x*(y*(x*y))}{x}{x*y}
%     \end{align*}
%However, note the variable $z$, representing an arbitrary witness, in $P_1[x, y]$.
Note the variable $z$, representing an arbitrary witness, in $P_1[x, y]$.
An arbitrary value is a correct witness in case  $x*y\simeq y*x$ holds, as in this case~\eqref{eq:spec-ex2} is trivially satisfied. Thus, we do not need to consider the case $x*y\simeq y*x$ separately. % and instead can return the same value for it as for the case $x*y\not\simeq y*x$.\footnote{More on program simplification in Section~\ref{sec:simplification}.}
Hence, we construct the final program $P[x,y]$ only from $P_2[x,y]$ and obtain: 
% Note the variable $z$, representing an arbitrary value, in the $\thencmd$ branch in the first line of the program.
% This is a valid witness for the case when $x*y\simeq y*x$ holds -- then specification~\eqref{eq:spec-ex2} is trivially satisfied.
% In practice we remove this $\thencmd$ branch in postprocessing and obtain the final program $P[x,y]$:
    \begin{align*}
        P[x,y] := &\ \ite{x\!*\!(y\!*\!x)\!\simeq\! x}{x}{(\ite{e\!\simeq\! x\!*\!(y\!*\!(x\!*\!y))}{x}{x\!*\!y})}
    \end{align*}

%describing an element $z$, square of which is $e$, given that the group operation is not commutative.
    
\end{example}

We conclude this section by illustrating the benefits of computable unifiers.
%% SHORT VERSION OF ex:inv2:
\begin{example}\label{ex:inv2}
Consider the  group theory  specification
\begin{equation}
    \forall x, y.\exists z.\ z*(i(x)*i(y)) = e, \label{eq:spec-ex3}
\end{equation}
describing the inverse element $z$ of $i(x)*i(y)$.
We annotate the inverse $i(\cdot)$ as uncomputable to disallow the trivial solution $i(i(x)*i(y))$.
Using computable unifiers, we synthesize\ifbool{shortversion}{}{\footnote{see derivation in Appendix~\ref{sec:appendix-examples}}} % the extended version~\cite{VampirSyntPrePrint} of our paper}
the program $y*x$; that is, a program computing $y*x$ as the inverse of $i(x)*i(y)$. 
\end{example}

\section{Computable Unification with Abstraction}\label{sec:uwa}
\begin{algorithm}[t]
\caption{Computable Unification with Abstraction\label{alg:uwa}}
\begin{tabbing}
  \newemph{\reserved{function} $\mathtt{mgu_{comp}}(E_1, E_2, E_3)$} \\
  \newemph{\quad \reserved{if} $E_3$ is uncomputable \reserved{then} fail} \\
  %\newemph{\quad\quad fail} \\
  \quad let $\mathcal{E}$ be a set of equations and $\theta$ be a substitution;\  $\mathcal{E} := \{E_1=E_2\}$; $\theta := \{\}$ \\
  %\quad let $\theta$ be a substitution; $\theta := \{\}$ \\
  \newemph{\quad let $\mathcal{D}$ be a set of disequalities; $\mathcal{D} := \emptyset$} \\
	\quad \reserved{repeat}\\
  \quad\quad \reserved{if} $\mathcal{E}$ is empty \reserved{then}\\
%  \quad\quad\quad \PH{TODO: if $\overline{r}$ itself contains any uncomputables, abstract them and extend $D$} \\
  \newemph{\quad\quad\quad \reserved{return} $(\theta, D)$ where $D$ is the disjunction of literals in $\mathcal{D}$} \\
  \quad\quad Select an equation $s=t$ in $\mathcal{E}$ and remove it from $\mathcal{E}$\\
  \quad\quad \reserved{if} $s$ coincides with $t$ \reserved{then} do nothing\\
  %\quad\quad\quad do nothing \\
  \quad\quad \reserved{else if} $s$ is a variable and $s$ does not occur in $t$ \reserved{then}\\
  \newemph{\quad\quad\quad \reserved{if} $s$ does not occur in $E_3$ or $t$ is computable \reserved{then} $\theta\!:=\!\theta\!\circ\!\{s\!\mapsto\!t\}; \mathcal{E}\!=\!\mathcal{E}\{s\!\mapsto\!t\}$} \\
  %\newemph{\quad\quad\quad\quad $\theta := \theta\circ\{s\mapsto t\};\ \mathcal{E}:=\mathcal{E}\{s\mapsto t\}$} \\
  \newemph{\quad\quad\quad \reserved{else if} $t = f(t_1,\dots,t_n)$ and $f$ is computable \reserved{then}} \\
  \newemph{\quad\quad\quad\quad $\theta\! :=\! \theta\!\circ\!\{s\!\mapsto\! f(x_1,\dots,x_n)\};\ \mathcal{E}\!:=\!\mathcal{E}\{s\!\mapsto\! f(x_1,\dots,x_n)\}\!\cup\!\{x_1\!=\!t_1,\dots,x_n\!=\!t_n\}$} \\
  \newemph{\quad\quad\quad\quad\quad where $x_1,\dots,x_n$ are fresh variables}\\
  \newemph{\quad\quad\quad \reserved{else if} $t = f(t_1,\dots,t_n)$ and $f$ is uncomputable \reserved{then} $\mathcal{D}:=\mathcal{D}\cup\{s \not\simeq t\}$} \\
  %\newemph{\quad\quad\quad\quad $\mathcal{D}:=\mathcal{D}\cup\{s \not\simeq t\}$}\\
  \quad\quad \reserved{else if} $s$ is a variable and $s$ occurs in $t$ \reserved{then} fail\\
  %\quad\quad\quad fail \\
  \quad\quad \reserved{else if} $t$ is a variable \reserved{then} $\mathcal{E}:=\mathcal{E}\cup\{t=s\}$\\
  %\quad\quad\quad $\mathcal{E}:=\mathcal{E}\cup\{t=s\}$ \\
  \quad\quad \reserved{else if} $s$ and $t$ have different top-level symbols \reserved{then} fail\\
%  \newemph{\quad\quad\quad \reserved{if} at least one of the top-level symbols of $s, t$ is uncomputable \reserved{then}} \\
%  \newemph{\quad\quad\quad\quad $\mathcal{D}:=\mathcal{D}\cup\{s \not\simeq t\}$} \\
%  \newemph{\quad\quad\quad \reserved{else} fail} \\
  %\quad\quad\quad fail \\
  \quad\quad \reserved{else if} $s\!=\!f(s_1,\dots,s_n)$ and $t\!=\!f(t_1,\dots,t_n)$ \reserved{then} $\mathcal{E}\!:=\!\mathcal{E}\!\cup\!\{s_1\!=\!t_1,\dots,s_n\!=\!t_n\}$
  %\quad\quad\quad $\mathcal{E}:=\mathcal{E}\cup\{s_1=t_1,\dots,s_n=t_n\}$
\end{tabbing}
\vspace*{-1em}
\end{algorithm}

When compared to the  \Sup{} calculus of Figure~\ref{fig:sup}, our extended \Sup{} calculus with answer literals from Figure~\ref{fig:sup-new} uses computable unifiers instead of mgus.
%In this section we strengthen our extended \Sup{} calculus to incorporate the benefits of mgu reasoning from standard \Sup. 
%For doing so,
To find computable unifiers, we introduce Algorithm~\ref{alg:uwa} by extending a standard  unification algorithm~\cite{Robinson65,DBLP:conf/ki/HoderV09} and an algorithm for unification with abstraction of~\cite{UWA-THI}.   Algorithm~\ref{alg:uwa} combines
 computable unifiers with mgu computation, resulting in the computable unifier $\theta:=\mathtt{mgu_{comp}}(E_1, E_2, E_3)$ to be further used in Figure~\ref{fig:sup-new}. 
% %\PH{We should be careful with the wording -- we want to present it as a theoretical result, because we actually don't use it in \vampire{} (blocking rule usage when the ordinary unifier would introduce uncomputable symbol into the answer literal is sufficient for solving all our examples -- and we should say this in the Implementation section).}
% To build a unification $\theta$ of a term $E_1$ with another term $E_2$ using the standard unification algorithm, we simultaneously traverse the tree of subterms of $E_1$ and $E_2$. If the subterm $s$ of $E_1$ and the subterm $t$ of $E_2$ have the same top-level functor, then we may continue with the unification of the operands of the functor. However, if $s$ and $t$ disagree on the top-level functor, or if $s$ is a variable contained in $t$, the standard unification algorithm fails. Finally, if $s$ is a variable not contained in $t$, we add the substitution $s \mapsto t$ to the unification $\theta$ and apply this unification on all subterms yet to be traversed. 
%

Algorithm~\ref{alg:uwa} modifies a standard unification algorithm to ensure computability of $E_3\theta$.
Changes compared to a standard unification algorithm are highlighted. %and discussed next. 
Algorithm~\ref{alg:uwa} does not add $s \mapsto t$ to $\theta$ if $s$ is a variable in $E_3$ and $t$ is uncomputable.
Instead, if $t$ is $f(t_1,\dots,t_n)$ where $f$ is computable but not all $t_1,\dots,t_n$ are computable, we extend $\theta$ by $s\mapsto f(x_1,\dots,x_n)$ and then add equations $x_1=t_1,\dots,x_n=t_n$ to the set of equations $\mathcal{E}$ to be processed.
Otherwise, $f$ is uncomputable and we perform an abstraction: we consider $s$ and $t$ to be unified under the condition that $s \simeq t$ holds.
Therefore we add a constraint $s \not\simeq t$ to the set of literals $\mathcal{D}$ which will be added to any clause invoking the computable unifier.
To discharge the literal $s \not\simeq t$, one must prove $s \simeq t$.
While $s$ can be later substituted for other terms, as long as we use %computable unification with abstraction
$\mathtt{mgu_{comp}}$, $s$ will never be substituted for an uncomputable term.
Thus, we conclude the following result. 

%The correctness of the algorithm is stated by the following theorem. Proof is in Appendix~\ref{sec:appendix-proofs}.

\newcommand\theoremfive{
Let $E_1, E_2, E_3$ be expressions.
Then $(\theta, D) := \mathtt{mgu_{comp}}(E_1, E_2, E_3)$ is a computable unifier.
% WHAT WOULD EVEN BE A DEFINITION OF "MOST GENERAL" FOR COMPUTABLE UNIFIERS?
%Further, it is the most general of such unifiers: if there exists a computable unifier $(\theta^\prime, D^\prime)$, then there exists a substitution $\tau$ such that $\theta^\prime = \theta\tau$ and $D^\prime\supseteq D\tau$.
}
\begin{theorem}\label{thm:2}
\theoremfive
\end{theorem}

% \section{Synthesizing the Program} % THIS BECAME THE LAST SUBSECTION OF SECT 3
% \PH{I'd rather make this the last subsection of Section~\ref{sec:anslits}.}

% \PH{TODO: Explain how we synthesize the program after concluding the proof by putting together programs with conditions.
% Example: how to go from if-then-else to no if-then-else -- instead of ITE we use unifiers w/ abstraction. (?)

% Finally, we conclude the proof search by deriving empty clause -- hence the last case of the program is the negation of all the other assumptions.}

% A PARAGRAPH ON SIMPLIFICATION IS INCLUDED IN THE IMPLEMENTATION SECTION (RIGHT BEFORE EXPERIMENTS)
%\section{Simplifying the Program}\label{sec:simplification}
%\input{simplification}

%\section{Effective Synthesis with AVATAR}\label{sec:avatar}
%\input{avatar}

% PREVIOUS VERSION:
%\section{Synthesis in Superposition (previous version)}
%\input{synthesis}

% \section{Inference Rules with Answer Literals (previous version)}
% \input{inference_rules}

% \section{Examples (previous version)}
% \input{examples}

\section{Implementation and Experiments}\label{sec:impl:exp}\label{sec:simplification}
\label{sec:experiments}\label{sec:avatar}
\paragraph{Implementation.} We implemented our saturation-based program synthesis approach in the 
 \vampire{}  prover~\cite{CAV13}. We used Algorithm~\ref{alg:saturation-new} with 
 the extended \Sup calculus of Figure~\ref{fig:sup-new}.
%
%modifications of the saturation algorithm and AVATAR, and selected rules of the modified superposition calculus\footnote{Note that \vampire{} does not yet use the unification with abstraction.
%All our examples -- including Example~\ref{ex:inv2} -- can be solved by guiding the proof search such that a rule with ordinary unifier $\theta$ is used only if $(\theta, \square)$ is a computable unifier.
% \PH{Note here that solving our examples actually does not require the UWA.
% It is sufficient to use regular unifier conditioned upon not introducing uncomputable symbols into answer literals.
% In a way, this is the simplest implementation of the search for computable unifiers: do ordinary unification, then check if the unifier $\theta$ introduces anything uncomputable into the expression $E_3$, and if it does, do not use it at all (i.e., block the rule application), otherwise $(\theta, \square)$ is a computable unifier.}
%}
%in the superposition-based theorem prover %\vampire{}~\cite{CAV13}.
%
%Our implementation is based on \vampire{}'s existing support for answer literals for question answering~\cite{Reger2018}, which we extended to synthesis.
The implementation, consisting of approximately 1100 lines of C++ code, is available  at \url{https://github.com/vprover/vampire}.
The synthesis functionality can be turned on using the option \texttt{--question\_answering synthesis}.

\vampire{} accepts functional specifications in an extension of the SMT-LIB2 format~\cite{SMTLIB-standard}, by using the new command \texttt{assert-not} to mark the specification. % (other \texttt{assert}s are treated as assumptions).
We consider  interpreted theory symbols %from the standard theories %(e.g., $+, <$)
to be computable.
Uninterpreted symbols can be annotated as uncomputable via  the command \texttt{(set-option\linebreak :uncomputable (symbol1 ... symbolN))}.

Our implementation also integrates Algorithm~\ref{alg:saturation-new} with the AVATAR architecture~\cite{AVATAR14}.\ifbool{shortversion}{}{\footnote{We include a short description of AVATAR in Appendix~\ref{sec:appendix-avatar}.}} %the extended version~\cite{VampirSyntPrePrint} of our paper.}
We modified the AVATAR framework to only allow splitting over ground computable clauses that do not contain answer literals.
Further, if we derive a clause $C[\overline{\sigma}]\lor\ans(r[\overline{\sigma}])$ with AVATAR assertions $C_1[\overline{\sigma}],\dots,C_m[\overline{\sigma}]$, where $C[\overline{\sigma}]$ is ground and computable, we replace it by the clause $C[\overline{\sigma}]\lor\bigvee_{i=1}^m \lnot C_i[\overline{\sigma}]\lor\ans(r[\overline{\sigma}])$ without any assertions.
We then immediately record a program with conditions $\prog{r[\overline{x}]}{\lnot C[\overline{x}]\land\bigwedge_{i=1}^m C_i[\overline{x}]}$, and replace the clause by $C[\overline{\sigma}]\lor\bigvee_{i=1}^m \lnot C_i[\overline{\sigma}]$ (see lines 7-10 of Algorithm~\ref{alg:saturation-new}), which may be then further split by AVATAR.

Finally, our implementation  simplifies  the programs we synthesize. 
% First, as mentioned in Section~\ref{sec:synth-with-al}, we only use a program with conditions recorded during the proof search if the condition appears in the derivation of $\square$.
% The reasoning for this is that we are looking for conditions $C_1[\overline{x}],\dots,C_n[\overline{x}]$ for the cases of our nested $\itecons$ as in~\eqref{eq:prog-construction} such that $\bigwedge_{i=1}^n A_i\land\bigwedge_{i=1}^k C_i[\overline{x}]$ is unsatisfiable.
% That we can guarantee by selecting only those programs from $\mathcal{P}$ whose conditions occur in the derivation of $\square$.
%Since these conditions cover all possible cases, we do not need to use any other programs from $\mathcal{P}$.
%Second,
If during Algorithm~\ref{alg:saturation-new} we record a program $\langle z, F\rangle$ where $z$ is a variable, %then we use it in the final program construction only if all other recorded programs also correspond to a variable
we do not use this program in the final program construction (line~12 of Algorithm~\ref{alg:saturation-new})  even if $F$ occurs in the derivation of $\square$ (see Example~\ref{ex:group-ite}).
%This is because $z$ represents an arbitrary witness -- therefore any other program $\langle t, F'\rangle$ (where $t$ is not a variable) computes a satisfactory witness also for the condition $F$.
%If all recorded programs correspond to a variable, then we choose a constant of the corresponding sort, and use that as the final program.
%\PH{Note: we could also simplify using the techniques from the interpolation paper or utilizing the simplification rules of our calculus. Should we mention that? (See comments in .tex below this.)}
% We could also employ simplification in the style of \PH{TODO: the simplification techniques from interpolation?}.
% Another option is to utilize the simplification rules of the theorem prover.
% In the saturation loop of Algorithm~\ref{alg:saturation-new} we can first simplify clause $C_i$, and only then carry out the program recording and answer literal removal on lines 6-10.
% %In practice we need to be careful with such simplification so that we do not modify the answer literal in $C_i$.
% %\PH{Do we though?}
% \PH{This is what we can do easily (and it's pretty much implemented, I'd just have to move one block of code to a different place in the saturation algorithm).
% Since we don't have the interpolation-style technique implemented, I would actually go with this. Thoughts?}

%\PH{Comment here on the practical obstacles? More rules than just the basic 5, proof possibly not straightforward.}
%\LK{stopped here}
\paragraph{Examples and experimental setup.}
The goal of our experimental evaluation is to showcase the benefits of our approach on problems that are deemed to be hard, even unsolvable,  by state-of-the-art synthesis techniques. We therefore focused on first--order theory reasoning and evaluated our work on the group theory problems of Examples~\ref{ex:inv}-\ref{ex:inv2}, as well as on integer arithmetic problems. 

%implementation can synthesize programs for other specifications, a selection of which we list here.
As the SMT-LIB2  format can  easily be translated into the SyGuS~2.1 syntax~\cite{sygus-standard}, we compared our results to % compare the performance of \vampire{}~4.7 (linked with Z3~4.9.1.0~\cite{Z3} for AVATAR) with 
\cvc5~1.0.4~\cite{cvc5}, supporting % predecessor \solver{cvc4}~\cite{CVC4} won in 4 out of 5 tracks in the most recent 
SyGuS-based synthesis ~\cite{sygus-comp}.
%Our goal is not to carry out a broad comparison, but rather to focus on specifics of the selected examples and showcase some differences and similarities between our approach and SyGuS.
Our experiments were run on an AMD Epyc 7502, 2.5 GHz CPU with 1 TB RAM, using a 5 minutes time limit per example. 
Our benchmarks as well as the configurations for our experiments are available at:\linebreak \url{https://github.com/vprover/vampire_benchmarks/tree/master/synthesis}

\paragraph{Experimental results with group theory properties.}
\vampire synthesizes the solutions of the Examples~\ref{ex:inv}-\ref{ex:inv2} in 0.01, 13, and 0.03 seconds, respectively.
Since these examples use uninterpreted functions, %these examples
they cannot be encoded in the SyGuS 2.1 syntax, showcasing the limits of other synthesis tools.

\paragraph{Experimental results with maximum of $n\geq 2$ integers.}
For the maximum of 2 integers, the specification is $\forall x_1, x_2\in\intg.\ \exists y\in\intg.\big(y\geq x_1 \land y\geq x_2 \land (y=x_1 \lor y=x_2)\big)$, and the program we synthesize is $\ite{x_1<x_2}{x_2}{x_1}$.
%Thanks to the integration with  AVATAR, 
Both our work and \cvc5 are able to synthesize  programs choosing the maximal value for up to $n=23$ input variables, as summarized below. For $n> 23$, both \vampire{} and \cvc5 time out.
%\cvc5 can also find the program for up to 23 variables.
%The following table shows times in seconds it took \vampire{} and \cvc5 to solve different versions of the benchmark:
\begin{center}\setlength{\tabcolsep}{8pt}
\begin{tabular}{r|ccccccc}
      Number $n$ of variables for  & \multirow{2}{*}{2}     & \multirow{2}{*}{5}     & \multirow{2}{*}{10}    & \multirow{2}{*}{15}    & \multirow{2}{*}{20}    & \multirow{2}{*}{22}    & \multirow{2}{*}{23} \\ 
      which max is synthesized & & & & & & &  \\ \hline
     \vampire{}             & 0.03  & 0.03  & 0.05  & 1     & 13    & 55    & 215\\
     \cvc5                  & 0.01  & 0.03  & 0.6   & 6.8   & 88    & 188   & 257\\  
\end{tabular}
\end{center}
%Neither of the solvers was able to synthesize a program finding a maximum of 24 variables within 5 minutes.

\paragraph{Experimental results with polynomial equations.}
\vampire can synthesize the solution of polynomial equations; for example, for  $\forall x_1, x_2 \in \intg.\exists y \in\intg. (y^2 = x_1^2 + 2x_1x_2 + x_2^2)$, we synthesize $x_1+x_2$.
\vampire{} finds the corresponding program in 26 seconds using simple first-order reasoning, while
\cvc5 fails in our setup. %even with the help of complex decision procedures for non-linear arithmetic.
%\PH{We're doing something wrong, Haniel was able to synthesize the program with \cvc5. TODO: Wait on his answer and then edit this part.}

\medskip
%\PH{We could also describe the workshop example from~\cite{Reger2018}, but we have no space.}
%\PH{Maybe also mention here the orderings where uncomputables have ordinal weight? Again, we haven't implemented it yet, but it is a nice idea for the future.}

% \begin{table}[]
%     \centering\setlength{\tabcolsep}{5pt}\renewcommand{\arraystretch}{1.2}
% \resizebox{\textwidth}{!}{
%     \begin{tabular}{ll}
%          Specification & Program  \\ \hline
%          Maximum of 2 integers: $\forall x_1, x_2\in\intg.\ \exists y\in\intg.$ \\
%          \quad $\big(y\geq x_1 \land y\geq x_2 \land (y=x_1 \lor y=x_2)\big)$ &  $\ite{x_1<x_2}{x_2}{x_1}$ \\ \hline
%          Polynomial equation:  \\
%          \quad $\forall x_1, x_2 \in \intg.\exists y \in\intg. (y^2 = x_1^2 + 2x_1x_2 + x_2^2)$ & $x_1+x_2$ \\
%     \end{tabular}}
%     \caption{}
%     \label{tab:examples}
% \end{table}

\section{Related Work}\label{sec:related}
Our work builds upon deductive synthesis~\cite{MannaWaldinger1980} adapted for the resolution calculus~\cite{LWC1974,Tammet1994}. We extend this line of work with saturation-based program synthesis, by using adjustments of  the superposition calculus.
% Manna and Waldinger~\cite{MannaWaldinger1980} present a deductive framework which combines theorem proving and program synthesis.
% The system uses unification, induction and transformation rules and the synthesised programs can contain recursive functions.
% There is no maintained implementation available.
% Tammet~\cite{Tammet1994} presents a complete $A(R)$-calculus, which, given a criteria $R$ of a certain type, finds a substitution $\{y\mapsto t\}$ for an existentially bound variable $y$ in a given formula such that $R(t)$ holds.
% The calculus uses answer literals, the synthesised programs supports recursion.

Component-based synthesis of recursion-free programs~\cite{StickelEtAl1994} from logical specifications is addressed in~\cite{StickelEtAl1994,GulwaniEtAl2011,TiwariEtAl2015}. 
The work of~\cite{StickelEtAl1994} uses first-order theorem proving  to prove  specifications and  extract  programs from  proofs. In~\cite{GulwaniEtAl2011,TiwariEtAl2015}, $\exists\forall$ formulas are produced to capture specifications over component properties and  SMT solving is applied to find a term satisfying the formula, corresponding to a straight-line program. We complement~\cite{StickelEtAl1994}  with saturation-based superposition proving and avoid template-based SMT solving from~\cite{GulwaniEtAl2011,TiwariEtAl2015}.
% Gulwani et al.~\cite{GulwaniEtAl2011} propose a constraint-based approach for solving component-based synthesis of loop-free programs.
% They encode the constraint as a $\exists\forall$ formula in FOL and present an algorithm for solving the constraint by counterexample guided iterative synthesis using SMT solvers.
% Tiwari et al.~\cite{TiwariEtAl2015} present an approach to component-based synthesis as well.
% Their method uses two distinct interpretations of program symbols -- one capturing functional constraints, the other one non-functional constraints.
% These constraints are then encoded as a $\exists\forall$ problem passed to the SMT solver Yices, which responds by finding a term satisfying the formula (if possible), corresponding to a straight-line program.

A prominent line of research comes with syntax guided synthesis (SyGuS)~\cite{DBLP:series/natosec/AlurBDF0JKMMRSSSSTU15}, where functional specifications are given using  a context-free grammar. This grammar yields program templates to be synthesized via an enumerative search procedure based on SMT solving~\cite{cvc5,ogis}. We believe our work is  complementary to SyGuS, by strengthening first-order reasoning for program synthesis, as evidenced by Examples~\ref{ex:inv}--\ref{ex:inv2}.
%Alur et al.~\cite{sygus-standard} introduced the Syntax Guided Synthesis (SyGuS) problem definition for program synthesis. In SyGuS, the program to be synthesized must derive from a context free grammar and satisfy a given SMT formula constraint.
%Solvers for SyGuS includes the SMT solver cvc5~\cite{cvc5}. %and DryadSynth~\cite{dryadsynth}.
%The cvc5~\cite{cvc5} SMT solver has support for this problem domain and performs well in the SyGuS competition~\cite{sygus-comp}.
%To solve a SyGuS problem, cvc5 repeatedly evaluates candidate solutions by invoking its SMT solver. As such it belongs to a large class of program synthesis techniques referred to as Oracle Guided Inductive Synthesis~\cite{ogis} techniques, where synthesis candidates are evaluated by external programs, providing feedback to refine the search. %If the oracle provides feedback in the form of counterexamples to refine the search, the method is referred to as a Counter-Example Guided Inductive Synthesis (CEGIS) method.

The sketching technique~\cite{sketch,rosette} %popular
%tools for program synthesis include Sketch~\cite{sketch} 
synthesizes program assignments to variables, using   % in a C-like programming language to pass assert statements, and Rosette~\cite{rosette}, an extension to the Racket programming language with high level meta-programming primitives. 
an alternative framework to the program synthesis setting we rely upon. In  particular, sketching  addresses  domains that do not involve input logical formulas as functional specifications,  such as example-guided synthesis~\cite{DBLP:conf/pldi/ThakkarNSANR21}.% However, these fall outside the scope of this paper. 

\section{Conclusions}\label{sec:conclusions}
We extend  saturation-based proof search to saturation-based program synthesis, aiming to derive recursion-free programs from specifications.
We integrate answer literals with saturation, and modify the superposition calculus and unification to synthesize computable programs.
Our initial experiments show that a first-order theorem prover becomes an efficient program synthesizer, potentially opening up interesting avenues toward recursive program synthesis, for example using saturation-based proving with induction.

\medskip

\paragraph{Acknowledgements.} 
We thank Haniel Barbosa for support with experiments with \cvc5.
This work was partially funded by the ERC CoG ARTIST 101002685, and the FWF grants LogiCS W1255-N23 and LOCOTES P 35787.

\bibliographystyle{splncs04}
\bibliography{bibliography}

\newpage
\appendix
\section{Proofs}\label{sec:appendix-proofs}

In the following we use the notion of \emph{universal closure} of a formula $F$, which is the formula $\forall \overline{z}. F$, where $\overline{z}$ are all free variables of $F$.

\setcounter{theorem}{0}
\begin{Theorem}
\theoremone
\end{Theorem}
\begin{proof}
We consider the calculus which was used for deriving $C\lor\ans(r[\overline{\sigma}])$, but with lifted ordering and selection constraints.
Since the soundness of the calculus does not depend on these constraints, the calculus without the constraints is sound as well.
Now, since $\ans$ is uninterpreted, we can replace $\ans(y)$ by $y\not\simeq r[\overline{\sigma}]$, and obtain a derivation of $C\lor r[\overline{\sigma}]\not\simeq r[\overline{\sigma}]$ from $A_1,\dots,A_n, \forall y.\cnf(\lnot F[\overline{\sigma}, y]\lor y\not\simeq r[\overline{\sigma}])$ using the calculus without the constraints.\footnote{The derivation might not have been possible in the calculus with the ordering and selection constraints due to replacing the positive literal $\ans(y)$ with the negative literal $y\not\simeq r[\overline{\sigma}]$ containing different symbols.}

We want to show that
\begin{equation}
\bigwedge_{i=1}^n A_i\land\bigwedge_{i=1}^m C_i \rightarrow C\lor F[\overline{\sigma}, r[\overline{\sigma}]]\label{eq:proof1}
\end{equation}
is valid.
Hence, we need to show that in each interpretation, in which the antecedent is true, also the consequent is true.
Let us consider such an interpretation $I$.
We distinguish two cases.
First, assume that $\forall y.\cnf(\lnot F[\overline{\sigma}, y]\lor y\not\simeq r[\overline{\sigma}])$ is true in $I$.
Then since all assumptions from which we derived $C\lor r[\overline{\sigma}]\not\simeq r[\overline{\sigma}]$ are true in $I$ and since the inference system is sound, also $C\lor r[\overline{\sigma}]\not\simeq r[\overline{\sigma}]$ is true.
That clause is equivalent to $C$, hence $C$ is true, which makes the consequent of~\eqref{eq:proof1} true.
Second, assume that $\forall y.\cnf(\lnot F[\overline{\sigma}, y]\lor y\not\simeq r[\overline{\sigma}])$ is false in $I$.
Then its negation, $\lnot \forall y.\cnf(\lnot F[\overline{\sigma}, y]\lor y\not\simeq r[\overline{\sigma}])$, equivalent to $\exists y. (F[\overline{\sigma}, y]\land y\simeq r[\overline{\sigma}])$, equivalent to $F[\overline{\sigma}, r[\overline{\sigma}]]$ must be true in $I$.
Hence, the consequent of~\eqref{eq:proof1} is true also in this case.
Therefore~\eqref{eq:proof1} is valid. \hfill $\square$
\end{proof}

\begin{Corollary}
\corollarytwo
\end{Corollary}
\begin{proof}
From Theorem~\ref{thm:1} follows that $\bigwedge_{i=1}^n A_i\land\bigwedge_{i=1}^m C_i[\overline{\sigma}] \rightarrow C[\overline{\sigma}]\lor F[\overline{\sigma}, r[\overline{\sigma}]]$ holds.
Since $\overline{\sigma}$ are fresh uninterpreted constants, we obtain that $\bigwedge_{i=1}^n A_i\land\bigwedge_{i=1}^m C_i[\overline{x}] \rightarrow C[\overline{x}]\lor F[\overline{x}, r[\overline{x}]]$ is valid as well, and that is equivalent to
$\bigwedge_{j=1}^{m}C_j[\overline{x}]\land\lnot C[\overline{x}] \rightarrow (\bigwedge_{i=1}^n A_i \rightarrow F[\overline{x}, r[\overline{x}]])$.
Therefore $\prog{r[\overline{x}]}{\bigwedge_{j=1}^{m}C_j[\overline{x}]\land\lnot C[\overline{x}]}$ is a program with conditions for $A_1\land\ldots\land A_n \rightarrow \forall \overline{x}.\exists y. F[\overline{x}, y]$.\hfill$\square$
\end{proof}

\setcounter{equation}{7}
\begin{Corollary}
\corollarythree
\end{Corollary}
\setcounter{equation}{13}
\begin{proof}
For any interpretation $I$ and any variable assignment $v$, let $p$ be the smallest index such that $\lnot C_p[\overline{x}]$ holds in $I$ under $v$, but all $\lnot C_j[\overline{x}]$, where $1\leq j<p$, do not hold in $I$ under $v$.
Since $\bigwedge_{i=1}^n A_i\land\bigwedge_{i=1}^k C_i[\overline{x}]$ is unsatisfiable, under the assumptions $A_1,\dots,A_n$ such a $p$ has to exist.
Then in $I$ under $v$ and under the assumptions $A_1,\dots,A_n$, the interpretation of $P[\overline{x}]$ is the same as the interpretation of $r_p[\overline{x}]$.

Further, since $\bigwedge_{j=1}^{p-1}C_j[\overline{x}]\land\lnot C_p[\overline{x}]$ is the condition for $P_p[\overline{x}]$, from the definition of a program with conditions we obtain that $A_1\land\dots\land A_n\rightarrow F[\overline{x}, r_p[\overline{x}]]$ holds in $I$ under $v$.
Hence also $A_1\land\dots\land A_n\rightarrow F[\overline{x}, P[\overline{x}]]$ holds in $I$ under $v$.

Finally, since this argument holds for any $I$ and $v$, and since all $A_1,\dots,A_n$ are closed formulas, also $A_1\land\dots\land A_n\rightarrow \forall \overline{x}. F[\overline{x}, P[\overline{x}]]$ holds.
Therefore $P[\overline{x}]$ is a program for~\eqref{eq:spec2}.\hfill $\square$
\end{proof}

\begin{Lemma}
\lemmathree
\end{Lemma}
\begin{proof}%[Lemma~\ref{lemma:soundness1}]
For clarity we repeat the original rule:
\begin{equation*}
    \infer[,]{C\theta}{C_1 \quad \cdots \quad C_n}\tag{\ref{eq:orig-rule}}
\end{equation*}
where $\theta$ is a substitution such that $E\theta\simeq E^\prime\theta$ holds for some expressions $E, E^\prime$.

We will prove the soundness of the new rule
\begin{equation*}
\infer[]{(D\lor C\lor \ans(r))\theta'}{C_1 \lor \ans(r) \quad C_2 \quad \cdots \quad C_n},\tag{\ref{eq:new-rules}'}
\end{equation*}
where $(\theta', D)$ is a computable unifier of $E, E^\prime$ with respect to $r$, and none of $C_1,\dots,C_n$ contains an answer literal.
The proof of soundness of the other new rules of~\eqref{eq:new-rules} is analogous.
%We define $\theta:=\theta'\cup\{x\mapsto t ~|~ x\not\simeq t \in D\}$.

Assume interpretation $I$ to be a model of the universal closures of the premises of~(\ref{eq:new-rules}'), but not a model of the universal closure of its conclusion.
Then $D\theta', C\theta'$ and $\ans(r)\theta'$ are false in $I$.
From $D\theta'$ being false in $I$ and from $(\theta', D)$ being an abstract unifier follows that $E\theta'\simeq E^\prime\theta'$ holds.
We can therefore set $\theta:=\theta'$.
From the soundness of~\eqref{eq:orig-rule} and $C\theta'$ being false in $I$ then follows that some of $C_1,\dots,C_n$ is false in $I$.
However, none of $C_2,\dots,C_n$ can be false in $I$, because we assumed all premises of~(\ref{eq:new-rules}') to be true in $I$.
Hence, $C_1$ is false in $I$.
Further, from $\ans(r)\theta'$ being false in $I$ follows that $\ans(r)$ is false in $I$.
However, that means that $C_1\lor\ans(r)$ is false in $I$, which contradicts the assumption that the universal closures of all premises of rule~(\ref{eq:new-rules}') are true in $I$.

Hence, the rule~(\ref{eq:new-rules}') is sound.\hfill $\square$
\end{proof}

\begin{Lemma}
\lemmafour
\end{Lemma}
\begin{proof}%[Lemma~\ref{lemma:soundness2}]
Soundness of the factoring, equality factoring and equality resolution rules follows from Lemma~\ref{lemma:soundness1}.

We will prove soundness for the first superposition rule and the second binary resolution rule.
The proofs for other superposition and binary resolution rules are analogous.

For clarity we repeat the first superposition rule of Figure~\ref{fig:sup-new}:
\[
\infer[]{(D\lor L[t]\lor C \lor C^\prime\lor\ans(\ite{s\!\simeq\! t}{r^\prime}{r}))\theta}{s\simeq t\lor C\lor \ans(r)\quad L[s^\prime]\lor C^\prime\lor\ans(r^\prime)}
\]
Assume interpretation $I$ to be a model of the universal closures of the premises of the rule, but not a model of the universal closure of its conclusion.
Then there is some variable assignment $v$ such that $(D\lor L[t]\lor C \lor C^\prime\lor\ans(\ite{s\!\simeq\! t}{r^\prime}{r}))\theta$ is false in $I$ under $v$.
Let $v^\prime$ be a variable assignment that assigns to each variable $x$ the value that $x\theta$ has in $I$ under $v$.
Then:
\begin{enumerate}
    \item From $L[t]\theta, C\theta, C^\prime\theta$ being false in $I$ under $v$ follows that $L[t], C, C^\prime$ are false in $I$ under $v^\prime$.
    \item Since $D\theta$ is false in $I$ under $v$, from $(\theta, D)$ being an abstract unifier of $s, s'$ follows that $s\theta\simeq s'\theta$ is true in $I$ under $v$, and therefore $s, s^\prime$ have the same interpretation in $I$ under $v^\prime$.
    Then consider two cases:
    \begin{enumerate}
        \item $s\simeq t$ is true in $I$ under $v^\prime$ and $s\theta\simeq t\theta$ is true in $I$ under $v$.
        Then from $\ans(\ite{s\!\simeq\! t}{r^\prime}{r})\theta$ being false in $I$ under $v$ follows that $\ans(r^\prime)\theta$ is false in $I$ under $v$ and therefore $\ans(r^\prime)$ is false in $I$ under $v^\prime$.
        Also from $s\simeq t$ being true in $I$ under $v^\prime$, 1., and 2. follows that $L[s^\prime]$ is false in $I$ under $v^\prime$.
        Then the whole second premise of the rule is false in $I$ under $v^\prime$, which is a contradiction with the assumption that $I$ is a model of its universal closure.
        \item $s\simeq t$ is false in $I$ under $v^\prime$ and $s\theta\simeq t\theta$ is false in $I$ under $v$.
        This case leads similarly to the first premise being false, in contradiction with the assumption.
    \end{enumerate}
\end{enumerate}
Therefore the first superposition rule is sound.

\smallskip

For clarity we repeat the second binary resolution rule of Figure~\ref{fig:sup-new}:
\[
\infer[]{(D\lor r\!\not\simeq\! r^\prime\lor C \lor C^\prime \lor \ans(r))\theta}
        {A\!\lor\! C\!\lor\! \ans(r)\quad \neg A^\prime\!\lor\! C^\prime\!\lor\! \ans(r^\prime)}
\]
Assume interpretation $I$ to be a model of the universal closures of the premises of the rule, but not a model of the universal closure of its conclusion.
Then there is some variable assignment $v$ such that $(D\lor r\!\not\simeq\! r^\prime\lor C \lor C^\prime \lor \ans(r))\theta$ is false in $I$ under $v$.
Let $v^\prime$ be a variable assignment that assigns to each variable $x$ the value that $x\theta$ has in $I$ under $v$.
Then:
\begin{enumerate}
    \item From $r\theta\not\simeq r^\prime\theta, C\theta, C^\prime\theta$ being false in $I$ under $v$ follows that $r\not\simeq r^\prime, C, C^\prime$ are false in $I$ under $v^\prime$. Therefore $r, r^\prime$ have the same interpretation in $I$ under $v^\prime$.
    \item Since $\ans(r)\theta$ is false in $I$ under $v$, also $\ans(r)$ is false in $I$ under $v^\prime$.
    Then from 1. follows that %since $r, r^\prime$ have the same interpretation in $I$ under $v^\prime$,
    $\ans(r^\prime)$ is also false in $I$ under $v^\prime$.
    \item Since $D\theta$ is false in $I$ under $v$, from $(\theta, D)$ being an abstract unifier of $A, A'$ follows that $A\theta, A^\prime\theta$ have the same interpretation in $I$ under $v$, and therefore $A, A^\prime$ have the same interpretation in $I$ under $v^\prime$.
    Therefore, only one of $A,\lnot A^\prime$ is true in $I$ under $v^\prime$, which together with $C, C^\prime, \ans(r), \ans(r^\prime)$ being false in $I$ under $v^\prime$ forms a contradiction with the assumption that $I$ is a model of both premises of the rule.
\end{enumerate}
Therefore the second binary resolution rule is sound as well.\hfill $\square$
\end{proof}

\begin{theorem}
\theoremfive
\end{theorem}
\begin{proof}
We will denote the subexpression of the expression $E$ at position $p$ by $E|p$.

We first prove that $(\theta, D)$ is an abstract unifier of $E_1, E_2$.
If $E_1\theta|p'$ and $E_2\theta|p'$ differ, there has to be a position $p$, where $p'$ is a prefix of $p$, such that the top-level symbol of $E_1\theta|p$ and $E_2\theta|p$ differs.
From the construction of $\theta$ follows that for any position $p$, the subexpressions $E_1\theta|p, E_2\theta|p$ differ in their top-level symbol only if $E_1|p = s$ and $E_2|p = f(t_1,\dots,t_n)$ (or, symmetrically, $E_1|p = f(t_1,\dots,t_n)$ and $E_2|p = s$) where $s$ is a variable and $f$ is uncomputable.
However, in this case $s\not\simeq f(t_1,\dots,t_n)$ occurs in $D$.
Therefore, for any interpretation $I$, any variable assignment $v$, and any position $p'$, the interpretations of $E_1\theta|p', E_2\theta|p'$ in $I$ under $v$ will either be the same, or $s\theta\not\simeq f(t_1,\dots,t_n)\theta$ will be true in $I$ under $v$.
Hence, $(D\lor E_1\simeq E_2)\theta$ is valid, and therefore $(\theta, D)$ is an abstract unifier of $E_1, E_2$.

Next, we prove that $E_3\theta$ is computable.
Since the algorithm successfully terminated, $E_3$ must have been computable (otherwise it would fail).
Further, the algorithm only extends the substitution $\theta$ by $s\mapsto t$ where $t$ is uncomputable if $s$ does not occur in $E_3$.
Thus, $E_3\theta$ is computable, and hence $(\theta, D)$ is a computable unifier.
% The fact that $(\theta, D)$ is most general follows from the construction.
% Except for the case when $E_1|p = s$ and $E_2|p = f(t_1,\dots,t_n)$, where $s$ is a variable not occurring in $f(t_1,\dots,t_n)$ but occurring in $E_3$ and $f(t_1,\dots,t_n)$ is uncomputable, the construction of $(\theta, D)$ corresponds to the standard unification algorithm constructing the most general unifier.
% In this special case, our algorithm distinguishes two cases.
% If $f$ is computable, then our algorithm ensures that the top-level symbol of term substituted for $s$ is $f$.
% Any unifier $\theta'$ has to do the same before recursively descending into subterms $t_1,\dots,t_n$, hence $\theta$ is not less general than $\theta'$.
% If $f$ is uncomputable, then our algorithm adds $s\not\simeq t$ to $D$, not extending $\theta$ at all.
\end{proof}

\section{Example}\label{sec:appendix-examples}
\setcounter{example}{2}
A detailed version of Example~\ref{ex:inv2} follows.
\begin{example}
%\PH{This example can use a computable unifier, but then it leads to a longer proof. Please choose whichever you find better.}
Consider the group axioms (A1)-(A3) of Figure~\ref{fig:groupax}, the additional axioms (A1')~$\forall x.\ x*i(x)\simeq e$ for right inverse and (A2')~$\forall x.\ x*e\simeq x$ for right identity (symmetric to (A1), (A2)),\footnote{We include the symmetric axioms to shorten the derivation for presentation purposes. The derivation would work also without (A1'), (A2').}
% \begin{equation}
%     \forall x.\ x*i(x)\simeq e, \tag{A1'}
% \end{equation}
and the following specification
\begin{equation}
    \forall x, y.\exists z.\ z*(i(x)*i(y)) = e,
\end{equation}
describing the inverse element of $i(x)*i(y)$.
%\PH{TODO: maybe remove one of the additional axioms depending on which proof we choose.}
The trivial program derivation for this specification would only have three steps:
{\footnotesize
\begin{enumerate}
    %12. e != op(X2,op(inv(sK2),inv(sK1))) | ~ans0(X2) [cnf transformation 11]; 1
    \item $e\not\simeq x*(i(\sigma_1)*i(\sigma_2))\lor\ans(x)$ \hfill [preprocessed specification]
    \item $\ans(i(i(\sigma_1)*i(\sigma_2)))$ \hfill [BR A1, 1.]
    \item $\square$ \hfill [answer literal removal 2.]
\end{enumerate}}%
To disallow the trivial solution, $i(i(x)*i(y))$, we annotate the function symbol $i$ as uncomputable.
%\PH{Here are two possible derivations. The first one uses a computable unifier, the second one just works around the uncomputability (Vampire finds the second one).}
Therefore we do not perform the step 2. from above, but instead perform binary resolution with the computable unifier $(\{x\mapsto i(i(\sigma_1)*i(\sigma_2))\}, x\not\simeq i(i(\sigma_1)*i(\sigma_2)))$, leading to the following derivation:
{\footnotesize
\begin{enumerate}
    %12. e != op(X2,op(inv(sK2),inv(sK1))) | ~ans0(X2) [cnf transformation 11]; 1
    \item $e\not\simeq x*(i(\sigma_1)*i(\sigma_2))\lor\ans(x)$ \hfill [preprocessed specification]
    %13. inv(op(inv(sK2),inv(sK1))) != X2 | ~ans0(X2) [cnf transformation 12]; 2
    \item $i(i(\sigma_1)*i(\sigma_2))\not\simeq x\lor\ans(x)$ \hfill [BR A1, 1.]
    %14. op(op(X1,X0),X2) = op(X1,op(X0,X2)) [cnf transformation 9]; A3
    %15. op(X0,e) = X0 [cnf transformation 3]; A2'
    %16. op(e,X0) = X0 [cnf transformation 2]; A2
    %17. e = op(inv(X0),X0) [cnf transformation 1]; A1
    %50. e = op(X2,inv(X2)) [superposition 17,41]; A1'
    %22. op(inv(X2),op(X2,X3)) = op(e,X3) [superposition 14,17]; 3
    \item $i(x)*(x*y) \simeq e*y$ \hfill [Sup A3, A1]
    %27. op(inv(X2),op(X2,X3)) = X3 [forward demodulation 22,16]; 4
    \item $i(x)*(x*y) \simeq y$ \hfill [Sup 3., A2]
    %32. op(inv(inv(X1)),e) = X1 [superposition 27,17]; 5
    \item $i(i(x))*e \simeq x$ \hfill [Sup 4., A1]
    %41. inv(inv(X1)) = X1 [superposition 32,15]; 6
    \item $i(i(x)) \simeq x$ \hfill [Sup 5., A2']
    %53. e = op(X1,op(X2,inv(op(X1,X2)))) [superposition 50,14]; 7
    \item $e \simeq x*(y*i(x*y))$ \hfill [Sup A1', A3]
    %91. op(X2,inv(op(X1,X2))) = op(inv(X1),e) [superposition 27,53]; 8
    \item $x*i(y*x) \simeq i(x)*e$ \hfill [Sup 4., 7.]
    %113. inv(X1) = op(X2,inv(op(X1,X2))) [forward demodulation 91,15]; 9
    \item $i(x) = y*i(x*y)$ \hfill [Sup 8., A2']
    %132. op(inv(X0),inv(X1)) = inv(op(X1,X0)) [superposition 27,113]; 10
    \item $i(x)*i(y) = i(y*x)$ \hfill [Sup 4., 9.]
    % 147. inv(inv(op(sK1,sK2))) != X2 | ~ans0(X2) [backward demodulation 13,132]; 11
    \item $i(i(\sigma_2,\sigma_1))\not\simeq x\lor\ans(x)$ \hfill [Sup 10., 2.]
    %148. op(sK1,sK2) != X2 | ~ans0(X2) [forward demodulation 147,41]
    %151. ~ans0(op(sK1,sK2)) [equality resolution 148]
    %152. ans0(X0) [answer literal]
    %12
    \item $\ans(\sigma_2*\sigma_1)$ \hfill [BR 6., 11.]
    %153. false [unit resulting resolution 152,151]
    \item $\square$ \hfill [answer literal removal 12.]
\end{enumerate}}%
We synthesize the program $P[x,y] := y*x$.
Note that there exists a different derivation only using computable unifiers in the form $(\theta, \square)$ (i.e., not using abstraction).
%\PH{Whichever derivation we choose, note here that there exists a different one.}
%Note that this derivation does not use abstraction -- there is a different and longer derivation that uses it.
\end{example}

\section{Splitting and AVATAR}\label{sec:appendix-avatar}
One of the keys to the efficiency of saturation-based theorem proving is \emph{clause splitting}, with the leading approach being the AVATAR architechture~\cite{AVATAR14}.
The main idea of splitting is as follows.
Let $S$ be a set of clauses and $C_1\lor C_2$ a clause such that $C_1, C_2$ have disjoint sets of variables.
We call such clauses $C_1, C_2$ the \emph{components} of $C_1\lor C_2$.
Then $S\cup \{C_1\lor C_2\}$ is unsatisfiable iff both $S\cup\{C_1\}$ and $S\cup\{C_2\}$ are unsatisfiable.
Therefore, instead of checking satisfiability of a set of large clauses, we  check the satisfiability of multiple sets of smaller clauses.
AVATAR implements this idea by using an interplay between a saturation-based first-order theorem prover and a SAT/SMT solver.
The SAT/SMT solver finds a set of clause components, satisfiability of which implies satisfiability of all split clauses.
These components, called \emph{assertions}, are then used by the theorem prover for further derivations in saturation.
All clauses derived using assertions $C_1,\dots,C_n$ are called \emph{clauses with assertion} $C_1,\dots,C_n$.

However, when using Algorithm~\ref{alg:saturation-new} for program synthesis with (standard) AVATAR to saturate a preprocessed specification~\eqref{eq:spec2ans}, we may  derive a clause $\ans(r[\overline{\sigma}])$ with assertions $C_1[\overline{\sigma}],\dots,C_m[\overline{\sigma}]$. By Theorem~\ref{thm:1}, we then obtain $$A_1,\dots,A_n,C_1[\overline{\sigma}],\dots,C_m[\overline{\sigma}]\vdash F[\overline{\sigma},r[\overline{\sigma}]].$$
If $C_1[\overline{\sigma}],\dots,C_m[\overline{\sigma}]$ are computable and ground, then $\prog{r[\overline{x}]}{\bigwedge_{i=1}^{m} C_i[\overline{x}]}$ is a program with conditions.
However, if not all of the assertions $C_1[\overline{\sigma}],\dots,C_m[\overline{\sigma}]$ are computable and ground, then Algorithm~\ref{alg:saturation-new} should continue reasoning with these assertions with the aim of reducing them to computable and ground literals.
This, however, is not directly possible in the AVATAR framework.
Hence, we modify AVATAR as described in Section~\ref{sec:experiments}.

\end{document}